\documentclass[mnsc,nonblindrev]{informs3}
\OneAndAHalfSpacedXI

\usepackage{natbib}
\bibpunct[, ]{(}{)}{,}{a}{}{,}
\def\bibfont{\small}

\TheoremsNumberedThrough
\ECRepeatTheorems
\EquationsNumberedThrough

\usepackage{amsmath,amsfonts,amssymb}  
\usepackage{setspace}                  
\usepackage{xspace}                    
\usepackage{changepage}                
\usepackage{graphicx}                  
\usepackage{thmtools,thm-restate}      

\allowdisplaybreaks

\usepackage[dvipsnames]{xcolor}
\usepackage[unicode,linktocpage=true,breaklinks,colorlinks,citecolor=blue,linkcolor=BrickRed]{hyperref}
\DeclareUnicodeCharacter{2217}{*}
\usepackage{pgfplots}
\pgfplotsset{compat=1.18}
\usepackage{tikz}
\usepackage{regexpatch}

\makeatletter
\xpatchcmd\thmt@restatable{%
\csname #2\@xa\endcsname\ifx\@nx#1\@nx\else[{#1}]\fi
}{%
\ifthmt@thisistheone
\csname #2\@xa\endcsname\ifx\@nx#1\@nx\else[{#1}]\fi
\else
\csname #2\@xa\endcsname[{Restated}]
\fi}{}{}

\makeatother

\newcounter{alg}
\theoremstyle{TH}

\newtheorem{Lemma}{Lemma}

\theoremstyle{EX}
\newtheorem{Definition}{Definition}

\makeatletter
\renewenvironment{proof}[1][Proof]{%
  \Trivlist\item[\hspace*{1em}\hskip\labelsep{\it #1.\enskip}]\ignorespaces}{%
  \leavevmode\unskip\penalty9999 \hbox{}\nobreak\hfill\Halmos\par
  \endTrivlist\addvspace{0pt}}
\makeatother

\newenvironment{subfigure}[2][c]{\begin{minipage}[#1]{#2}}{\end{minipage}}

\providecommand{\normalsizeXI}{\normalsize}
\providecommand{\normalsizeX}{\normalsize}
\providecommand{\smallX}{\small}
\providecommand{\footnotesizeIX}{\footnotesize}

\newcommand{\OPT}{\texttt{OPT}\xspace}
\newcommand{\ALG}{\texttt{ALG}\xspace}

\newenvironment{exampleblock}
  {\begin{adjustwidth}{2em}{2em}
   \begin{spacing}{0.9}
   \vspace{0.5em}}
  {\end{spacing}
   \end{adjustwidth}
   \vspace{0.5em}}

\newenvironment{GeneralTheoremRestated}
  {\begin{repeattheorem}[Theorem~\ref{Thm1} (Restated).]}
  {\end{repeattheorem}}

\makeatletter
\def\AppendixFontSize{\normalsize}
\newcommand{\AppendixNumbering}{%
  \setcounter{equation}{0}%
  \@addtoreset{equation}{section}%
  \renewcommand{\theequation}{\thesection.\arabic{equation}}%
  \def\theHsection{appendix.\Alph{section}}%
  \def\theHequation{appendix.\Alph{section}.\arabic{equation}}}
\makeatother

\MANUSCRIPTNO{}
\hypersetup{
  pdftitle={The Power of Flexible Budgets in Adwords},
  pdfauthor={Suho Kang and Rajan Udwani},
  pdfsubject={Online advertising and budget flexibility},
}

\begin{document}
\RUNAUTHOR{Kang and Udwani}
\RUNTITLE{The Power of Flexible Budgets in Adwords}
\TITLE{The Power of Flexible Budgets in Adwords}
\ARTICLEAUTHORS{%
\AUTHOR{Suho Kang, Rajan Udwani}
\AFF{Department of Industrial Engineering and Operations Research (IEOR),
University of California, Berkeley,
\EMAIL{suho\_kang@berkeley.edu, rudwani@berkeley.edu}}
}

\ABSTRACT{%
Search advertising platforms routinely spend beyond an advertiser's average daily budget on high-traffic days, so long as total spending over the month stays within the monthly budget. Motivated by this practice, we study a $D$-day generalization of the Adwords problem \citep{mehta2007adwords}, where each advertiser $i$ has a nominal (average) daily budget $B_i$ and a total horizon (monthly) budget $DB_i$. Given a flexibility parameter $\delta$, the platform may spend at most $\delta B_i$ on advertiser $i$ on any single day, subject to the horizon spending limit of $DB_i$.

We quantify the power of $\delta$-flexible budgets by benchmarking against the inflexible offline optimum, which may spend at most $B_i$ on advertiser $i$ on each day. We show that no amount of flexibility helps direct generalizations of the classical algorithm of \citet{mehta2007adwords}. By contrast, for every fixed $\delta$, we design an algorithm whose competitive ratio converges to $1-e^{-\delta}$ as $D\to\infty$, and we show that this is asymptotically optimal. Perhaps surprisingly, this matches the optimal competitive ratio in a more permissive setting where the algorithm receives a fresh spending limit of $\delta B_i$ each day and may spend up to $\delta D B_i$ over the horizon. Along the way, we characterize the exact optimal competitive ratio for every pair $(D,\delta)$ on high-traffic instances, where the offline benchmark exhausts every advertiser's budget on every day.
}
\KEYWORDS{Online Advertising, Budget Flexibility,
    Online Algorithms, Competitive Ratio}
\maketitle

\section{Introduction}
\label{SecIntroduction}

Online advertising is a vast market: U.S. internet advertising revenue alone
reached hundreds of billions in 2024 \citep{IABAnnualReport2025}. At this scale, even
small improvements in the efficiency of ad allocation translate into
substantial value for platforms and advertisers alike.

To capture online ad allocation in search advertising, \citet{mehta2007adwords}
introduce the canonical Adwords problem. A set of advertisers $I=[n]$ and their
daily budgets $(B_i)_{i\in I}$ are known at the start of the day. Search queries arrive sequentially and (nonnegative) bids
$(b_{it})_{i\in I}$ are revealed upon the arrival of query $t$. The platform must immediately and
irrevocably show at most one ad, i.e., assign the query to at most one advertiser. Assigning query $t$ to
advertiser $i$ charges $i$ an amount $b_{it}$, and the total amount spent by
advertiser $i$ over the day may not exceed its budget $B_i$. The objective is to
maximize total spending across all advertisers.

In the classical formulation, the platform has no information about future queries or bids, and the arrival sequence may be adversarial. We evaluate the platform's (online) algorithm \ALG against the offline optimum
\OPT, which knows the entire sequence of queries and bids in advance. For
$\rho\in[0,1]$, \ALG is \emph{$\rho$-competitive} if $\ALG\geq\rho\,\OPT$ on
every instance, and the competitive ratio of \ALG is the largest such $\rho$.

 \citet{mehta2007adwords} analyze this problem in the \emph{small-bid regime}, where
\[
    \gamma:=\max_{i, t}\frac{b_{it}}{B_i}\to0,
\]
which we also assume throughout the paper. In this regime, they show that the optimal competitive ratio is $1-1/e$ and that it is achieved by \emph{reduced-bid} algorithm that discounts each
advertiser's bid according to the fraction of its budget already spent, and
awards the query to the advertiser with the largest discounted bid. 

Formally, given a penalty function $f:[0,1]\to[0,1]$, let $x_i(t)$ denote the
fraction of advertiser $i$'s budget spent before query $t$ arrives. Upon the
arrival of $t$, the corresponding reduced-bid algorithm computes
\[
b_{it}\bigl(1-f(x_i(t))\bigr),
\]
for every advertiser and assigns $t$ to a maximizer, discarding $t$ if all
reduced bids are zero. The penalty function $f$ is chosen to be non-decreasing, so advertisers that have already spent a larger fraction of their budgets become less attractive. Further, $f(1)=1$, so that once an advertiser exhausts its budget, its reduced bid becomes zero. \citet{mehta2007adwords} propose \textsc{Balance}, which uses
the exponential penalty
\[
    f(x)=e^{x-1}.
\]

With a rigid daily spending limit, \ALG may exhaust some advertisers' budgets
too early in the day while leaving others underspent---a misallocation that
\OPT, knowing the query sequence in advance, avoids. Permitting \ALG to spend beyond the nominal daily budget on some days and to
compensate by underspending on other days may improve overall performance.

In practice, platforms exercise exactly this flexibility, subject to a strict
longer-horizon cap, which we refer to as the \emph{horizon spending limit}.
Google Ads, for example, permits spending of up to twice an advertiser's nominal
daily budget on any single day, while capping monthly spending at $30.4$ times
that budget\footnote{The average month contains $365/12\approx30.4$ days.}
\citep{GoogleAdsSpendingLimits}.
As Figure~\ref{fig:GoogleAdwords} illustrates, the platform may spend more on
high-traffic days and less on others, provided the horizon spending limit is
respected. Presumably, flexibility is bounded because advertisers generally
prefer their budgets to be paced over the horizon rather than consumed in a few
days.

\begin{figure}[h]
    \centering
\includegraphics[width = 1\textwidth]{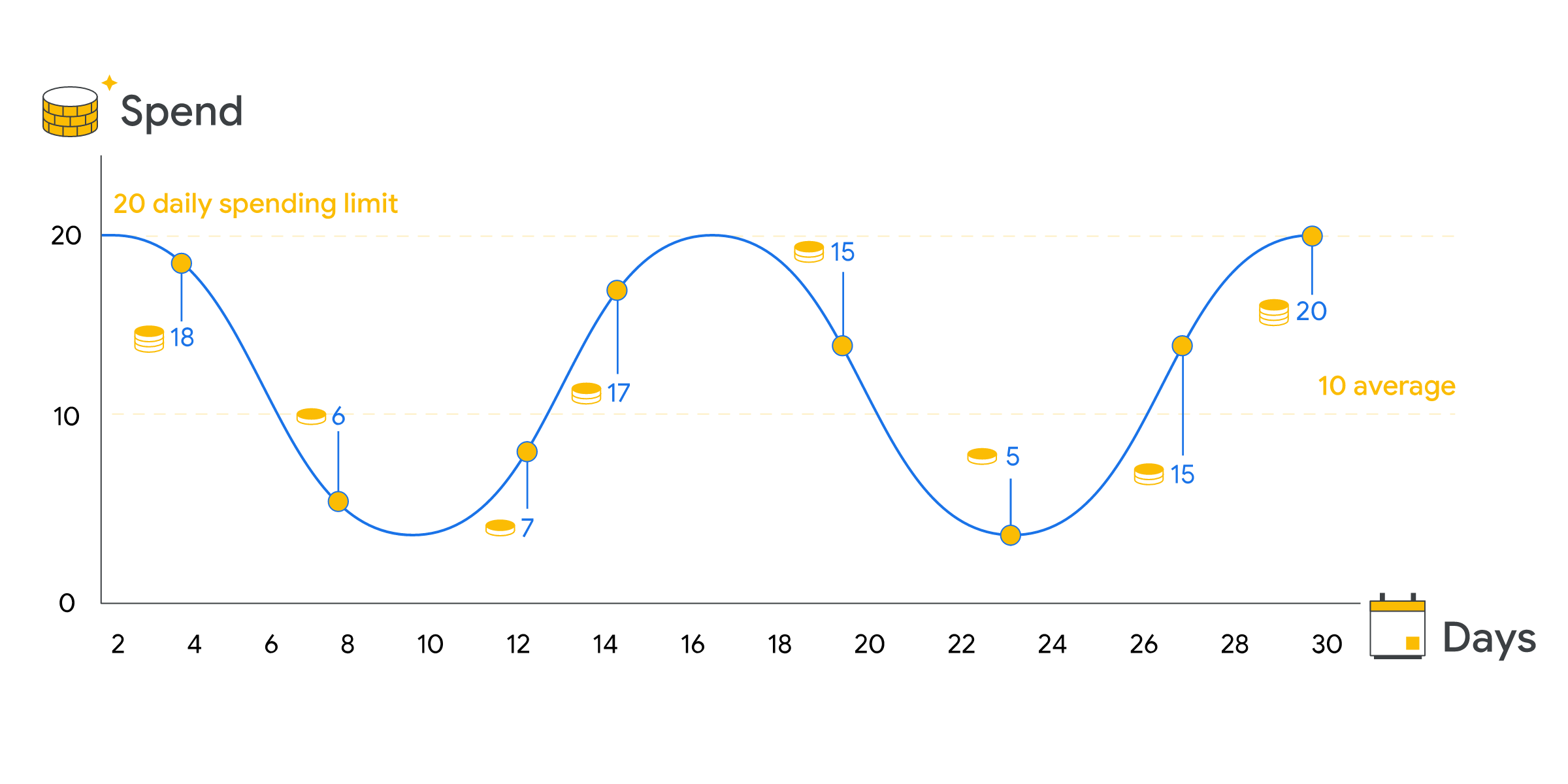}
    \caption{An illustration from Google AdWords intuitively explaining the flexible daily spending limit. The horizon spending limit is defined over approximately $30$ days, while the daily spending limit allows up to twice the nominal daily budget, corresponding to a flexibility level of $2$ in our model.}
    \label{fig:GoogleAdwords}
\end{figure}

The horizon spending limit makes any such gain far from obvious because
spending above the nominal daily budget on one day leaves less available for
future days. This tension motivates the central question of this paper:
\begin{quote}
\emph{Does budget flexibility improve the worst-case performance of online
algorithms relative to the inflexible offline benchmark? If so, how does the
improvement depend on the amount of flexibility?}
\end{quote}
\vspace{0.5em}

\subsection{Our Setting: $D$-day Adwords with $\delta$-flexibility}

Let $D$ denote the number of days in the campaign and let $\delta\in[1,D]$ denote the level of flexibility. As before, a set of advertisers $I=[n]$ is given, where each advertiser $i\in I$ has a nominal daily budget $B_i$ and a horizon budget $DB_i$ over $D$ days. 

On each day $d\in[D]$, an unknown (adversarial) sequence of queries $Q_d$ arrives online. Upon the arrival of query $t\in Q_d$, its bid vector $(b_{it})_{i\in I}$ is revealed, and an online algorithm \ALG must assign the query to at most one advertiser. For each advertiser $i$, \ALG may spend up to $\delta B_i$ per day, while the offline optimum \OPT, which knows the entire sequence of queries and bids in advance, may spend up to $B_i$ per day. Both \ALG and \OPT are subject to the same horizon spending limit $DB_i$. Our goal is to design an online algorithm with the best possible competitive ratio.

\subsection{Main Results}\label{MainResultsIntro}

Consider the following thought experiment: suppose that the online algorithm does not face the horizon spending limit and instead receives a fresh budget $\delta B_i$ for every advertiser $i \in I$, every day. Since \OPT is restricted to $B_i$ on each day, we can decompose the entire instance into $D$ single-day instances in which the online algorithm has $\delta$ times the budget of the offline optimum. In this setting, it can be shown that a natural generalization of \textsc{Balance} with the penalty function $f_\delta(x) := \exp(x-\delta)$ achieves the optimal competitive ratio of $1 - \exp(-\delta)$ (see Appendix \ref{SecFreshDailyBudget} for details). Since we removed the horizon spending limit in this thought experiment, $1-\exp(-\delta)$ also serves as an upper bound on the optimal competitive ratio in our main model.

Encouraged by this, we first consider the same generalization of \textsc{Balance} in our setting (with the horizon spending limit). Interestingly, additional flexibility buys this algorithm nothing and  for every $D$ and $\delta$, we construct an instance on which the algorithm does not beat $1- 1/e$ of the offline optimum (see Appendix~\ref{UBclassicBALANCE}). The reason is that the exponential penalty is too greedy
early on. The algorithm may spend aggressively on a single
advertiser on early days and quickly exhaust that advertiser's entire horizon
budget, leaving the other advertisers' budgets largely unspent. The offline
optimum instead spreads the same advertiser's budget evenly across days. This indicates a need for more conservative penalty on earlier days.

Guided by this observation, we first restrict attention to a structured class of instances in which \OPT exhausts every advertiser's nominal daily budget on every day. We call this the \emph{perfect-utilization setting}, which captures a high-traffic regime where demand is sufficient to fully utilize all daily budgets. Fixing the daily spending of \OPT\ allows us to focus entirely on the online algorithm's budget allocation across days. We refer to the model without this restriction as the \emph{general setting}.

In the perfect-utilization setting, we introduce Algorithm~\ref{alg:LPA}, a new reduced-bid algorithm whose penalty function on each day depends on both the remaining horizon budget and the number of days left. Algorithm~\ref{alg:LPA} represents a family of algorithms for different values of $D$ and $\delta$ and is formally introduced in Section~\ref{SecModelAlgorithm}. On the last day, it reduces to the exponential penalty of the single-day thought experiment. However, with more days remaining, it is more conservative and generally takes a form that is no longer a simple exponential. We first show that Algorithm~\ref{alg:LPA} achieves the optimal competitive ratio in the perfect-utilization setting for every $D$ and $\delta$.

\begin{restatable}{Theorem}{TheoremTwo} \label{Thm2}
For every $D$ and $\delta$, in the perfect-utilization setting, Algorithm~\ref{alg:LPA} has competitive ratio
\begin{equation}\label{cr}
    \frac{1}{D}
    \mathbb{E}\!\left[
        \min\left\{
            D,\sum_{d=1}^{D}\min\{Z_d,\delta\}
        \right\}
    \right], 
\end{equation}
where $Z_1,\ldots,Z_D$ are independent exponential random variables with mean one. This is the optimal competitive ratio under perfect-utilization, i.e., no online algorithm has competitive ratio strictly higher than \eqref{cr}.
\end{restatable}

Theorem~\ref{Thm2} has a natural probabilistic reading: each day contributes an independent mean-one exponential amount, truncated at the daily spending limit $\delta$, and the sum of these contributions is truncated again at the horizon limit $D$. Because the perfect-utilization setting restricts the instance space, its optimal competitive ratio is at least that of the general setting; consequently, Theorem~\ref{Thm2} also yields a finite-horizon upper bound for the general setting, strictly tighter than $1-\exp(-\delta)$ whenever $\delta>1$. We next analyze the same algorithm in the general setting, without imposing the perfect-utilization restriction. 

\begin{restatable}{Theorem}{TheoremOne} \label{Thm1}
For every $D$ and $\delta$, in the general setting, Algorithm~\ref{alg:LPA} achieves a competitive ratio of at least
\[
    \mathbb{E}\!\left[
        \min\left\{
            D,\sum_{d=1}^{D}\min\{Z_d,\delta\}
        \right\}\right]
        -
        \mathbb{E}\!\left[\min\left\{
            D,\sum_{d=1}^{D-1}\min\{Z_d,\delta\}
        \right\}
    \right],
\]
where $Z_1,\ldots,Z_D$ are independent exponential random variables with mean one.
\end{restatable}

The resulting guarantee is exactly the marginal expected consumption contributed by one additional capped mean-one exponential random variable, subject to the horizon spending limit. Since each capped exponential has mean $1-\exp(-\delta)<1$, the probability that their cumulative sum reaches the horizon limit vanishes as $D$ grows. Therefore, the competitive ratio approaches $1-\exp(-\delta)$ as $D$ grows. This, perhaps surprisingly, implies that Algorithm \ref{alg:LPA} has the same asymptotic guarantee as the optimal guarantee in the thought experiment where the online algorithm faces no horizon spending limit at all. In other words, with the right algorithm, the horizon spending limit is asymptotically costless, and Algorithm \ref{alg:LPA} is asymptotically optimal even in the general setting.

\begin{restatable}{Corollary}{CorAsymptotic}\label{CorAsymptotic}
    Fix $\delta\geq1$. As $D\to\infty$, the competitive ratio of Algorithm~\ref{alg:LPA} for the $D$-day Adwords problem with $\delta$-flexibility converges to $1-\exp(-\delta)$.
\end{restatable}

It remains to understand the cost of the perfect-utilization assumption when $D$ is finite. Unfortunately, unlike in the
classical single-day Adwords problem, the  perfect-utilization restriction is not without loss of generality. The general setting is strictly ``harder'' and may have worst case scenarios that violate perfect-utilization, as the following
non-asymptotic result shows.

\begin{restatable}{Theorem}{TheoremThree}\label{Thm3}
For $D=\delta=2$, the optimal competitive ratio in the general setting is
\[
    1-e^{-h^*}\approx0.701,
\]
where $h^*\approx1.208$ is the unique solution to
\[
    (1+h^*)e^{h^*-2}=1.
\]
\end{restatable}

This ratio is strictly below the perfect-utilization ratio at $D=\delta=2$,
which Theorem~\ref{Thm2} gives as approximately $0.729$ (see Figure \ref{fig:three_theorems_delta2}). Algorithmically, the
optimal first-day penalty in the general setting blends the penalty of
Algorithm~\ref{alg:LPA} with the exponential penalty. This is natural, since the
general setting must simultaneously handle (i) instances with no second-day
queries, where preserving budget for the future is pointless and the exponential
penalty is appropriate, and (ii) perfect-utilization instances, which call for a
more conservative penalty. As the remaining horizon budget decreases, the optimal
penalty accordingly shifts from the exponential form toward the
perfect-utilization penalty of Algorithm~\ref{alg:LPA}. Extending this exact
characterization beyond $D=2$, however, requires identifying and balancing a
growing family of candidate penalties that does not readily yield to analysis.

This modification is necessary for optimality at finite $D$, but its benefit
vanishes as $D$ grows: By Corollary~\ref{CorAsymptotic}, Algorithm~\ref{alg:LPA} is
asymptotically optimal. Therefore, the perfect-utilization setting does not always
pin down the exact finite-horizon worst case but it identifies the right
algorithmic structure for the large-horizon regime and supplies a tighter
finite-horizon upper bound for the general setting.

\begin{figure}[!htbp]
\centering
\begin{tikzpicture}
\begin{axis}[
    width=0.97\linewidth,
    height=6.1cm,
    xmode=log,
    log basis x=10,
    xmin=1.85, xmax=108,
    ymin=0.57, ymax=0.90,
    xtick={2,3,5,10,30,100},
    xticklabels={2,3,5,10,30,100},
    ytick={0.60,0.65,0.70,0.75,0.80,0.85,0.90},
    xlabel={Number of days $D$},
    ylabel={Competitive ratio (CR)},
    grid=major,
    major grid style={draw=gray!18},
    tick label style={font=\small},
    label style={font=\small},
    legend style={
        at={(0.5,-0.23)}, anchor=north,
        legend columns=2,
        draw=none, fill=none,
        font=\scriptsize
    },
    legend cell align=left,
]
\addplot[red!75!black,very thick,mark=*,mark size=2pt]
coordinates {
    (2,0.5939941503)
    (3,0.6157293749)
    (5,0.6293980674)
    (10,0.6675935984)
    (30,0.7589961651)
    (100,0.8469653495)
};
\addlegendentry{Lower bound on CR of Algorithm~\ref{alg:LPA} (Theorem~\ref{Thm1})}

\addplot[blue!65!black,very thick,mark=square*,mark size=2pt]
coordinates {
    (2,0.7293294335)
    (3,0.7650905800)
    (5,0.7989697335)
    (10,0.8303601295)
    (30,0.8561861027)
    (100,0.8641051876)
};
\addlegendentry{Optimal CR under perfect utilization (Theorem~\ref{Thm2})}

\addplot[black,only marks,mark=star,mark size=4pt]
coordinates {(2,0.7011878105)};
\addlegendentry{Optimal CR in the general setting (Theorem~\ref{Thm3})}

\addplot[gray!75!black,thick,dashed,mark=none]
coordinates {(2,0.8646647168) (100,0.8646647168)};
\addlegendentry{Thought experiment: $1-e^{-2}$ (Corollary~\ref{CorAsymptotic})}
\end{axis}
\end{tikzpicture}
\caption{Summary of the main competitive ratio bounds for $\delta=2$ at selected horizon lengths. The lower red curve shows the finite-horizon competitive ratio
guarantee certified by the LP-free analysis of Algorithm~\ref{alg:LPA} (Theorem~\ref{Thm1}), approaching $1-\exp(-2) \approx 0.865$ as $D$ grows. 
In particular, when $D=30$, Algorithm~\ref{alg:LPA} is at least $0.759$-competitive in the general setting, breaking the classical $1-1/e\approx0.632$ barrier with an approximately $20\%$ relative improvement in the guarantee for a monthly horizon. The upper blue curve is the optimal competitive ratio
under perfect utilization from Theorem~\ref{Thm2}, and hence an
upper bound for the general setting. Theorem~\ref{Thm3} identifies
the exact optimal competitive ratio in the general setting at $D=2$, approximately $0.701$,
strictly between these two bounds.
}
\label{fig:three_theorems_delta2}

\end{figure}
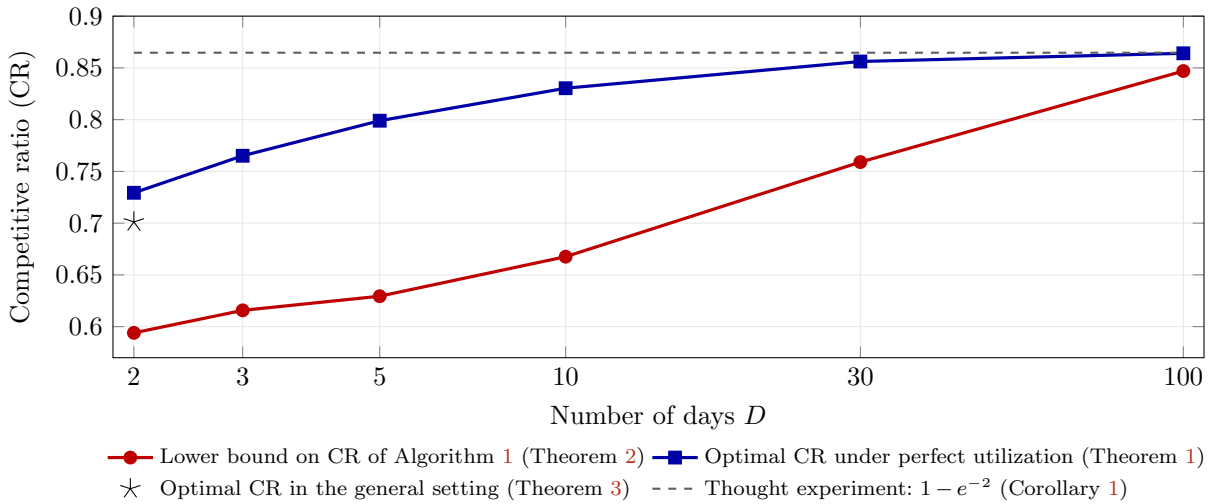

\paragraph{Technical Overview.}

Our analysis contends with three main difficulties: (i) the daily spending limits
are coupled across days through the horizon budget and past usage; (ii) the
appropriate penalty function therefore varies with both the remaining horizon
budget and the number of remaining days; and (iii) in the general setting the
offline usage is no longer fixed, so the analysis must track both online and
offline usages across days.

A direct application of the standard primal-dual analyses for Adwords
(e.g., \citealp{buchbinder2007online,devanur2013randomized}) is less transparent in our multi-day setting. The standard LP relaxation of the inflexible
offline benchmark decomposes across days into $D$ structurally identical linear
programs, one per day. Classical dual-fitting certifies a competitive ratio
$\alpha$ by exhibiting, for every advertiser $i$ and query $t$ on every day,
dual variables satisfying $\lambda_t+b_{it}\theta_i/B_i \geq\alpha\,b_{it}$. However, our algorithm's advantage over $1-1/e$ emerges only once gains
are amortized across the entire horizon. Fitting a dual
inequality for each query makes this approach cumbersome in our setting (see Appendix \ref{apdLPFreeJustification}). The LP-free framework of \citet{goyal2025asymptotically} builds this
amortization directly into the dual, replacing the per-query constraint with a
single constraint per advertiser. This makes it easy to define the candidate dual and lets us focus on designing the
right penalty functions and verifying that this condition holds with the
desired approximation factor. Concretely, we use a direct generalization of the
dual fitting of \citet{devanur2013randomized}, reducing the amortized analysis above to a \emph{factor-revealing
program}  over the
admissible online and offline usages of a \emph{single} advertiser. The same
certificate underlies all of our results, so the remaining work lies in
the design and analysis of the penalty functions.

Our next task is to choose penalty functions that maximize the worst-case value
of this program. Motivated by \citet{mehta2007adwords}, we first consider the more restricted perfect-utilization setting,  where \OPT exhausts every advertiser's nominal daily budget on every day, so that the offline usages are fixed at one and only the online usages remain.  We then work backward from the
last day. On the last day, the exponential
penalty of the thought experiment is optimal. At a high level, when moving backward, we choose each
day's penalty so that the current day's contribution, together with the worst
case over the remaining days, is the same no matter how much the algorithm spends
today. We call this the \emph{equalization property}. The common
value it produces corresponds to the expected sum of one capped mean-one
exponential per remaining day, truncated at the remaining budget, which is
exactly the guarantee in Theorem~\ref{Thm2}. This also gives us a reinterpretation of the classic $(1-1/e)$ guarantee as the expected value of a capped exponential random variable, and we refer to this as the \emph{exponential reinterpretation}. 
Applying the two properties recursively defines the optimal penalty functions (see Algorithm~\ref{alg:LPA}) and simultaneously
proves a lower bound on the guarantee. We show a matching upper bound by constructing a recursive generalization of the classical
upper-triangular hard instance. Together these prove Theorem~\ref{Thm2}.

We next analyze the algorithm in the general setting to prove its asymptotic optimality. At this point, the penalty functions are fixed, but the factor-revealing program again depends on both the online and offline usages. Here, using the equalization property, we eliminate the dependency on the \emph{online} usages from the factor-revealing program, now leaving only the offline usages. This step is almost a mirror image of before: we fixed the offline usages
to design the penalties, then remove the online usages to analyze the effect of variable offline usages. 
We show that the resulting guarantee is exactly the marginal expected consumption
contributed by one additional capped mean-one exponential
random variable, subject to the horizon spending limit, proving Theorem~\ref{Thm1} and Corollary~\ref{CorAsymptotic}.

Finally, we show that restricting to perfect utilization can strictly increase the optimal competitive ratio relative to the general setting. This contrasts with the classical Adwords problem \citep{mehta2007adwords}, where the guarantee derived under perfect utilization extends to arbitrary offline utilization without any loss. For $D=\delta=2$, we return to the original factor-revealing program and reformulate it based on the extreme cases of offline usages. These cases lead to different constraints, each point to a different family of penalty functions. The main difficulty is therefore to construct a single penalty function that balances these competing candidates, which appears to be quite involved for longer horizons.

\subsection{Related Work}
To the best of our knowledge, this is the first paper to study budget
flexibility from a worst-case perspective. Although the practice appears to be standard among advertising platforms, we are not aware of prior work that
quantifies its effect on the competitive ratio of online allocation. Several
lines of work are nonetheless closely related, and we discuss them below. Given the rich history and vast literature on online matching and Adwords, we focus on settings that are most directly related to ours and refer readers to \citet{mehta2010online,choi2020online,huang2024online, niazadeh2026modern} for comprehensive surveys and tutorials, including on recent developments. 

 Beginning with the online bipartite matching problem, which corresponds to binary bids and unit advertiser budgets, \citet{karp1990optimal} show that their randomized algorithm \texttt{RANKING} achieves the optimal competitive ratio of $1-1/e$ under an adversarial arrival sequence. Subsequently, \citet{kalyanasundaram2000optimal} consider the online $b$-matching problem, in which each bid is binary and all advertisers start with the same budget $b$, and show that the same optimal guarantee can be achieved by a deterministic algorithm as $b\to\infty$. \citet{mehta2007adwords} extend this to the Adwords problem in the small-bid
regime through \textsc{Balance}. \citet{buchbinder2007online} give a primal-dual
analysis of \textsc{Balance}, and \citet{devanur2013randomized} develop a
randomized primal-dual framework that recovers and unifies these classical
results. More recently, the LP-free framework of \citet{goyal2025asymptotically}, which we use, compares online allocation problems directly with the offline
benchmark without an explicit LP relaxation.

There have been several related directions that modify the classical Adwords problem to capture practical features beyond the standard formulation. One line of work relaxes the arrival model. For example,
\citet{devanur2009adwords,manshadi2012online,mirrokni2012simultaneous,devanur2019near}
obtain improved guarantees under more structured arrivals, including random-order
and stochastic models. We retain adversarial arrivals throughout. Our setting is also related to models in which resource availability changes over time. \citet{mehta2007adwords} show that \textsc{Balance} is robust to advertisers entering at different times. \citet{feng2025robustness} consider a setting where inventory is replenished according to an exogenous process, developing a batched generalization of the inventory balancing algorithm that retains the asymptotically optimal competitive ratio. \citet{kang2025black} further develop a black-box approach for exogenous replenishment in broader settings. However, the horizon budget in our model is fixed at the beginning, and no additional budget is introduced over time. Flexibility only changes how much of this fixed horizon budget may be spent on each day.

Closest in spirit to our work is a line that
identifies structure that breaks the $1-1/e$ barrier under adversarial arrivals. \citet{abolhassani2022online} show that in vertex-weighted online bipartite
matching the achievable guarantee improves as supply exceeds demand, and
\citet{barrientos2025online} obtain higher competitive ratios, under both
adversarial and stochastic arrivals, as the market becomes more imbalanced. 
\citet{manshadi2025online} study online matching with multichannel traffic and show that the competitive ratio improves with the fraction of degree-one arrivals, surpassing $1-1/e$ when this fraction is large. \citet{feng2025batching} consider settings in which online arrivals occur in multiple batches and show that the optimal competitive ratio is greater than $1-1/e$ by designing penalty functions that change across batches. Beyond linear objectives, \citet{devanur2012online} consider online matching with concave returns and characterize the optimal competitive ratio via differential equations, showing that strict concavity can improve the optimal guarantee beyond $1-1/e$. In a similar vein, \citet{aouad2023online} use the concavity of objective in online two-sided assortment optimization to derive algorithms with guarantee better than $1-1/e$.   We share the spirit of these
papers, however, the improvements resulting from budget flexibility are not (to our knowledge) directly related to supply-demand imbalance, concavity of the objective, or the degree and granularity of arrivals. 

Another related line of work varies what the algorithm knows about budgets:
\citet{vazirani2026towards} and
\citet{udwani2025adwords} consider Adwords with initially unknown budgets where the algorithm has strictly less information than in the classical
setting. In a different direction, \citet{choo2026learning} augment
\textsc{Balance} with predictions, trading off consistency when the advice is
accurate against worst-case robustness when it is not. Our algorithm knows every budget in advance and receives no advice.

Finally, our results are confined to the small-bid regime. Without that
assumption, \citet{kapralov2013online} show that no online algorithm is better
than $0.612$-competitive relative to the natural LP relaxation of Adwords, and
\citet{huang2024adwords} give a $0.5016$-competitive algorithm for general bids,
the first improvement over the greedy algorithm. Understanding the power of
budget flexibility beyond the small-bid regime is an interesting direction for
future work.
\paragraph{Outline.}
We first introduce the main algorithm (Algorithm \ref{alg:LPA}) in Section~\ref{SecModelAlgorithm}. Section~\ref{SecFramework} develops the LP-free framework and analyzes the perfect-utilization setting to derive the penalty functions of Algorithm~\ref{alg:LPA} and prove Theorem~\ref{Thm2}. Section~\ref{SecMainResult} then analyzes the general setting, establishing the asymptotic optimality result in Theorem~\ref{Thm1} and the exact two-day characterization in Theorem~\ref{Thm3}. Finally, Section~\ref{SecConclusion} concludes the paper and summarizes the main implications of our results.

\section{Horizon-Aware Reduced-Bid Algorithm}
\label{SecModelAlgorithm}

For ease of notation, we normalize all budget-related quantities of each advertiser $i\in I$ by its nominal daily budget $B_i$. Thus, each advertiser starts with a normalized horizon budget $D$, the daily spending limit of \ALG is $\delta$, and the daily spending limit of \OPT is $1$. Also, for $x \geq 0$, define
\[
    c(x):=\min\{x,\delta\}.
\]
In particular, if $r$ is the remaining normalized horizon budget of an advertiser at the beginning of a day, then $c(r)$ is its effective daily spending limit on that day.

A simple baseline algorithm for our setting is to ignore the flexibility and run classical \textsc{Balance} independently on each day, restricting the daily spending from every advertiser to $1$. Since \OPT is also subject to the same daily spending limit, the classical $1-1/e$ guarantee applies separately to each day and hence over the entire horizon.
Thus, this baseline algorithm is $(1-1/e)$-competitive for any $D$, and this guarantee is optimal if flexibility is ignored.\footnote{For example, if there are no arrivals after the first day, i.e., $Q_2=\cdots=Q_D=\emptyset$, the problem reduces to the single-day Adwords problem when flexibility is ignored. Therefore, the classical $1-1/e$ upper bound \citep{kalyanasundaram2000optimal,mehta2007adwords} applies.}

This baseline algorithm can be extended to the thought experiment introduced in Section~\ref{SecIntroduction}, where \ALG is only subject to a daily spending limit of $\delta$, without a horizon spending limit. In particular, applying the exponential penalty $f_\delta(x)=\exp(x-\delta)$ every day achieves the optimal competitive ratio of $1-\exp(-\delta)$. We refer readers to Appendix~\ref{SecFreshDailyBudget} for the proof.

This thought experiment also identifies the optimal penalty function on the last day of our main setting. If an advertiser enters the last day with remaining normalized horizon budget $r$, its effective spending limit is simply $c(r)$ and we can reduce the problem of designing the optimal penalty function for the last day to the single-day setting, for which the exponential penalty $f_{c(r)}(x)=e^{x-c(r)}$ is optimal.

To fully utilize flexibility in our setting, a natural first attempt is to use the exponential penalty on each day, with $c(r)$ determined by the remaining horizon budget at the beginning of that day. However, the following example shows that doing so throughout the horizon can be too aggressive and can even perform strictly worse than the baseline.

\begin{exampleblock}
\noindent\textbf{Example 1.} Consider $D=\delta=2$ and fix an arbitrarily small $\epsilon\in(0,1)$. There are two advertisers with nominal daily budgets
\[
    B_1=B,
    \qquad
    B_2=\frac{1+e^{-2}}{1-e^{-2}}B,
\]
where $B$ is sufficiently large. On day $1$, the first $2B$ queries are compatible with both advertisers. Advertiser $1$ bids one on every query, while advertiser $2$ bids
\[
    b_{2t}
    =
    \frac{(1-\epsilon)
    \left(1-e^{\frac{t-1}{B}-2}\right)}
    {1-e^{-2}},
    \qquad t\in[2B].
\]

Under the exponential penalty, \ALG prefers advertiser $1$ for these queries, as we verify below. After these arrivals, $B$ additional unit-bid queries arrive on day $1$ and another $B$ on day $2$, each compatible only with advertiser $1$.

\medskip
\noindent
\emph{Offline Optimum.} \OPT assigns the first $2B$ queries to advertiser $2$ and the remaining queries to advertiser $1$, which is feasible under its inflexible daily spending limits. Hence, as $B\to\infty$,
\[
    \frac{\OPT}{B}
    \geq
    2+\frac{1}{B}\sum_{t=1}^{2B}b_{2t}
    \longrightarrow
    2+(1-\epsilon)\frac{1+e^{-2}}{1-e^{-2}}.
\]

\medskip
\noindent
\emph{Online Algorithm.}
Suppose that the algorithm uses the exponential penalty
\[
    f_{\delta}(x)=e^{x-2}
\]
on the first day. Before the $t$-th query for $t \in [2B]$, advertiser $1$ has normalized usage $(t-1)/B$, while advertiser $2$ remains unused. Their reduced bids are therefore
\[
    1-e^{\frac{t-1}{B}-2}
    \qquad\text{and}\qquad
    (1-\epsilon)
    \left(1-e^{\frac{t-1}{B}-2}\right),
\]
respectively. Thus, the first $2B$ queries are assigned to advertiser $1$, exhausting its entire horizon budget on the first day, and
\[
    \ALG=2B.
\]
Letting $B\to\infty$ and then $\epsilon\to0$ gives
\[
    \frac{\ALG}{\OPT}
    \leq
    \frac{2(1-e^{-2})}{3-e^{-2}}
    \approx0.604
    <
    1-\frac1e.
\]
\end{exampleblock}

In Appendix~\ref{UBclassicBALANCE}, we extend this construction and show that, for every $D\geq 2$ and $\delta\in[1,D]$, a direct generalization of \textsc{Balance} cannot guarantee a competitive ratio greater than $1-1/e$. Although the exponential penalty is optimal on the last day, the above example suggests that the algorithm should be more ``conservative'' on earlier days, imposing a larger penalty on early spending to preserve budget for future arrivals.

At the same time, being uniformly conservative across days is also not desirable. For example, if on day $d$ there is at least $\delta (D-d+1)$ (normalized) budget available then we are essentially in the though-experiment regime and no budget needs to be preserved---the optimal penalty is again the exponential penalty. Therefore, the penalty should depend on both the remaining horizon budget and the number of remaining days. 

Motivated by these observations, we consider a \emph{horizon-aware reduced-bid algorithm}, whose penalty functions depend on both the remaining horizon budget and the number of remaining days. To formalize this dependence, at the beginning of day $d$, let $R_{id}$ denote the normalized remaining horizon budget of advertiser $i$. Upon the arrival of query $t\in Q_d$, let $x_{id}(t)$ denote the normalized amount already spent from advertiser $i$ on day $d$. Given a penalty function $f(\cdot\mid R_{id},d)$, which will be specified shortly, upon the arrival of query $t\in Q_d$, the algorithm computes advertiser $i$'s reduced bid as
\[
    b_{it}
    \left(
        1-f\left(x_{id}(t)\mid R_{id},d\right)
    \right),
\]
and assigns the query to an advertiser with the largest positive reduced bid. If no advertiser has a positive reduced bid, the query is discarded. Ties are broken arbitrarily. 

For any normalized remaining budget $r\geq0$ and day $d\in [D]$, we require $f(\cdot\mid r,d)$ to satisfy the following conditions.

\begin{Definition}[Valid Penalty Functions]
    Given $r>0$ and $d\in[D]$, a function
    \[
        f(\cdot\mid r,d):[0,c(r)]\to[0,1]
    \]
    is valid if:
    \begin{enumerate}
        \item $f(c(r)\mid r,d)=1$.
        \item $f(\cdot\mid r,d)$ is non-decreasing.
        \item $f$ is $C$-Lipschitz for some fixed constant $C>0$.
    \end{enumerate}
\end{Definition}

The first condition ensures that, once advertiser $i$ reaches its effective daily spending limit, its reduced bid becomes zero and the algorithm no longer assigns additional queries to that advertiser on the same day. The second condition ensures that the penalty does not decrease as more budget is spent from the advertiser on that day. The third is a regularity condition ensuring that the penalty does not change substantially over the small amount of budget consumed by a single query.\footnote{The classical exponential penalty and all penalty functions used in our lower-bound arguments satisfy this condition with $C=1$.}

To specify our choice of penalty functions, let
$Z_1,Z_2,\ldots$ be independent exponential random variables with mean one,
and define
\[
    \nu_\ell(x)
    :=
    \Pr\left(
        \sum_{k=1}^{\ell} c(Z_k)\geq x
    \right),
    \qquad
    \nu_0(x):=0.
\]

\begin{Definition}[Horizon-Aware Reduced-Bid Algorithm]
\label{DefLPF}
\refstepcounter{alg}\label{alg:LPA}
Given a remaining normalized horizon budget $r\geq0$ and the current day
$d\in[D]$, define
\begin{equation}\label{LPF}
    f(x\mid r,d)
    :=
    e^{x-c(r)}
    +
    \int_x^{c(r)}
        e^{x-z}\nu_{D-d}(r-z)\,dz,
    \qquad
    x\in[0,c(r)].
\end{equation}
We refer to the reduced-bid algorithm that uses these penalty functions as Algorithm~\thealg. Also, when needed, we write $f^{D,\delta}(x\mid r,d)$ instead to make the dependence on $D$ and $\delta$ explicit.
\end{Definition}

We note that the resulting penalty functions are indeed valid (see Appendix~\ref{apd_LPFproperties}). The first term in $f(\cdot \mid r, d)$ is the exponential penalty from $f_{c(r)}$, while the second term provides an additional correction for the remaining days. As the name ``horizon-aware'' suggests, even for the same normalized remaining horizon budget, the penalty function may change depending on the number of remaining days. The precise form of this correction originates from our analysis in Section~\ref{SecFramework}.

In the classical Adwords problem, where $D=1$, the only penalty function required is $f^{1,1}(\cdot\mid 1,1)$, and \citet{mehta2007adwords} shows that $f^{1,1}(x\mid 1,1)=\exp(x-1)$ is optimal. In general, if an advertiser is left with normalized remaining horizon budget $r$ on the last day, Algorithm~\ref{alg:LPA} chooses a penalty function of
\[
    f^{D, \delta}(x\mid r,D)=f_{c(r)}(x) = e^{x-c(r)},
\]
which matches the exponential penalty in the thought experiment. The following example illustrates how the additional correction term modifies this exponential penalty on earlier days, and how the resulting penalty depends on both the day and the remaining horizon budget.

\begin{exampleblock}
\noindent\textbf{Example 2.} Consider the case \(D=3\) and \(\delta=2\), and fix an advertiser, starting with a horizon budget of \(3\). On the first day, Algorithm~\ref{alg:LPA} uses the penalty function \(f^{3, 2}(x\mid3,1)\),  which is\footnote{The piecewise form appears because the correction term in \eqref{LPF} changes when the remaining horizon budget crosses the threshold \(\delta=2\).}
\[
f^{3, 2}(x\mid3,1)=
\begin{cases}
\left(e+5-2x-\dfrac{x^2}{2}\right)e^{x-3},
    & 0\le x\le1,\\[0.8em]
\left(e+6-4x+\dfrac{x^2}{2}\right)e^{x-3},
    & 1\le x\le2.
\end{cases}
\]

Now suppose we are at the beginning of the second day, and let \(x_1\) denote the normalized amount spent from the advertiser on the first day. To illustrate how the penalty changes with the remaining horizon budget, consider \(x_1=0,1,\) and \(2\). The corresponding remaining budgets are \(3,2,\) and \(1\), and the second-day penalty functions \(f^{3, 2}(x\mid3-x_1,2)\) are
\[
f^{3, 2}(x\mid3,2)=
\begin{cases}
(e+1)e^{x-3}, & 0\le x\le1,\\[0.3em]
(e+2-x)e^{x-3}, & 1\le x\le2,
\end{cases}
\]
\[
f^{3, 2}(x\mid2,2)=(3-x)e^{x-2},
    \qquad 0\le x\le2,
\]
and
\[
f^{3,2}(x\mid1,2)=(2-x)e^{x-1},
    \qquad 0\le x\le1\footnote{Since the remaining horizon budget is $1$, the effective daily spending limit is $c(1)=1$.}.
\]

Finally, if \(x_2\) denotes the normalized amount spent on the second day, the remaining horizon budget on the last day is \(3-x_1-x_2\). Since there are no future days, the correction term vanishes and the penalty reduces to 
\[
    f^{3, 2}(x\mid3-x_1-x_2,3)
    =
    e^{x-c(3-x_1-x_2)}, \qquad 0\le x\le c(3-x_1 - x_2).
\]

Figure~\ref{fig:exampleD3delta2} illustrates the dependence of the penalty on both the day and the remaining horizon budget. The left panel shows that, for the same remaining horizon budget \(r=3\), the penalty is more conservative earlier in the horizon and approaches the exponential form \(\exp(x-2)\) on the last day. The right panel shows that, for the same day \(d=2\), the penalty becomes more conservative as the remaining horizon budget decreases.
\end{exampleblock}

\begin{figure}[h]
\centering

\begin{minipage}{0.47\textwidth}
\centering
\begin{tikzpicture}
\begin{axis}[
    width=\textwidth,
    height=0.72\textwidth,
    xmin=0, xmax=2.05,
    ymin=0, ymax=1.05,
    xlabel={$x$},
    ylabel={$f(x\mid r,d)$},
    xtick={0,1,2},
    ytick={0,0.2,0.4,0.6,0.8,1},
    samples=200,
    thick,
    legend style={
        at={(0.98,0.02)},
        anchor=south east,
        draw=none,
        fill=none,
        font=\tiny
    },
    clip=false
]

\addplot[red, domain=0:1, forget plot] {exp(x-3)*(exp(1)+5-2*x - x^2/2)};
\addplot[red, domain=1:2] {exp(x-3)*(exp(1)+6-4*x + x^2/2)};
\addlegendentry{$f^{3, 2}(x\mid 3,1)$}

\addplot[blue, domain=0:1, forget plot] {exp(x-3)*(exp(1)+1)};
\addplot[blue, domain=1:2] {exp(x-3)*(exp(1)+2-x)};
\addlegendentry{$f^{3, 2}(x\mid 3,2)$}

\addplot[green!60!black, domain=0:2] {exp(x-2)};
\addlegendentry{$f^{3,2}(x\mid 3,3)$}

\draw[densely dashed, gray] (axis cs:1,0) -- (axis cs:1,1.02);

\end{axis}
\end{tikzpicture}

\vspace{0.3em}
\small (a) Fixed \(r=3\), compare days \(1,2,3\)
\end{minipage}
\hfill
\begin{minipage}{0.47\textwidth}
\centering
\begin{tikzpicture}
\begin{axis}[
    width=\textwidth,
    height=0.72\textwidth,
    xmin=0, xmax=2.05,
    ymin=0, ymax=1.05,
    xlabel={$x$},
    ylabel={$f(x\mid r,2)$},
    xtick={0,1,2},
    ytick={0,0.2,0.4,0.6,0.8,1},
    samples=200,
    thick,
    legend style={
        at={(0.98,0.02)},
        anchor=south east,
        draw=none,
        fill=none,
        font=\tiny
    },
    clip=false
]

\addplot[red, domain=0:1] {exp(x-1)*(2-x)};
\addlegendentry{$f^{3, 2}(x\mid 1,2)$}

\addplot[blue, domain=0:2] {exp(x-2)*(3-x)};
\addlegendentry{$f^{3, 2}(x\mid 2,2)$}

\addplot[green!60!black, domain=0:1, forget plot] {exp(x-3)*(exp(1)+1)};
\addplot[green!60!black, domain=1:2] {exp(x-3)*(exp(1)+2-x)};
\addlegendentry{$f^{3, 2}(x\mid 3,2)$}

\addplot[black, dotted, domain=0:2] {exp(x-2)};
\addlegendentry{$e^{x-2}$}

\draw[densely dashed, gray] (axis cs:1,0) -- (axis cs:1,1.02);

\end{axis}
\end{tikzpicture}

\vspace{0.3em}
\small (b) Fixed \(d=2\), compare remaining budgets \(r=3,2,1\)
\end{minipage}

\caption{Penalty functions in Example~2 for \(D=3\) and \(\delta=2\). The left panel shows that the penalty is more conservative earlier in the horizon: $f^{3,2}(x \mid 3, 1)$ lies above $f^{3,2}(x \mid 3, 2)$, which in turn lies above the last-day exponential penalty $f^{3,2}(x \mid 3, 3) = \exp(x-2)$. The right panel fixes the day to be $2$ and shows that a smaller remaining horizon budget leads to a more conservative penalty.}
\label{fig:exampleD3delta2}
\end{figure}
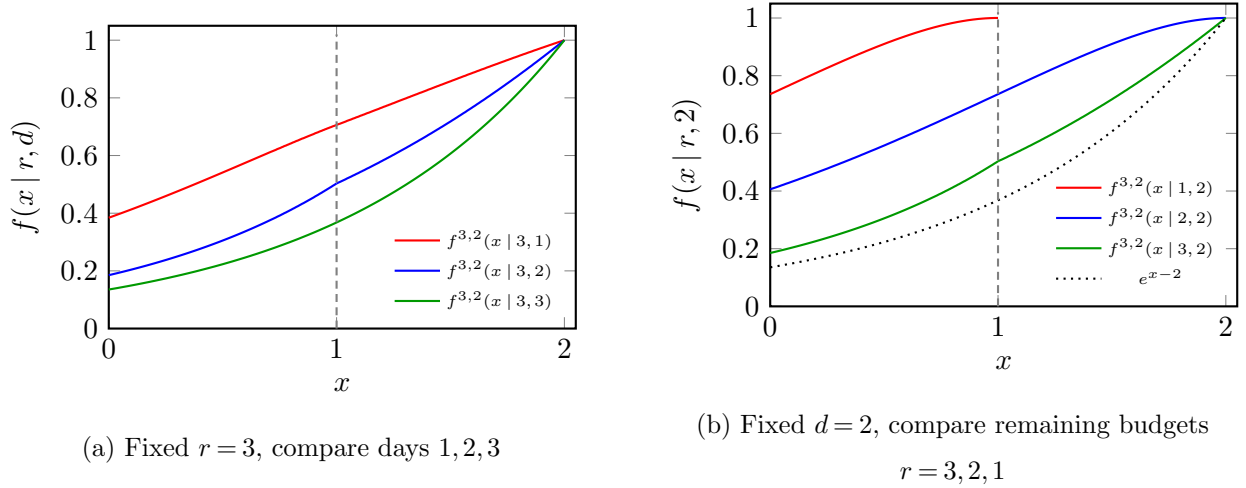

\section{Derivation and (Special Case) Analysis of Algorithm~\ref{alg:LPA}}
\label{SecFramework}

To analyze a reduced-bid algorithm, we start with some notation that is useful for comparing an online algorithm \ALG with the offline optimum \OPT. Recall that we also use \ALG and \OPT to denote their respective objective values. For each advertiser $i\in I$ and each day $d\in[D]$, let $X_{id} \in [0, \delta]$ denote the normalized amount spent by \ALG from advertiser $i$ on day $d$, and let $O_{id} \in [0, 1]$ denote the corresponding normalized amount spent by \OPT. Recall that $R_{id}$ denotes the normalized remaining horizon budget of advertiser $i$ at the beginning of day $d$. In terms of the total normalized budget used before day $d$,
\[
    R_{id}
    =
    D-\sum_{\ell<d}X_{i\ell}.
\]
and the objective values of \ALG and \OPT can be recovered from these normalized usages as
\[
    \ALG
    =
    \sum_{i\in I}\sum_{d\in[D]} B_iX_{id},
    \qquad
    \OPT
    =
    \sum_{i\in I}\sum_{d\in[D]} B_iO_{id}.
\]

Since our competitive ratio analysis must hold for every realization of the online and offline usages, we define the following sets of admissible usage vectors.

\begin{Definition}[Admissible Usages]
The set of admissible online usages is defined as
\[
    S_{\ALG}
    :=
    \left\{
        (x_1,\ldots,x_D)
        \,\middle|\,
        x_d\in[0,\delta]\ \forall d\in[D],
        \quad
        \sum_{d\in[D]}x_d\leq D
    \right\},
\]
and the set of admissible offline usages is defined as
\[
    S_{\OPT}
    :=
    \left\{
        (o_1,\ldots,o_D)
        \,\middle|\,
        o_d\in[0,1]\ \forall d\in[D],
        \quad
        \sum_{d\in[D]}o_d>0
    \right\}.
\]
\end{Definition}

We note that a direct application of the standard primal-dual approaches for Adwords (e.g., \citealp{buchbinder2007online, devanur2013randomized}) is less transparent in our multi-day setting. At a high level, these analyses construct a dual solution that satisfies the dual constraints locally for each advertiser and each query. In our setting, however, the penalty function and its domain depend on the remaining horizon budget and the number of remaining days, so the corresponding local guarantees naturally vary across days. To improve beyond the standard $1-1/e$ guarantee, this suggests amortizing these guarantees over the entire horizon for each advertiser, which naturally leads to the LP-free analysis of \citealp{goyal2025asymptotically}. We refer readers to Appendix~\ref{apdLPFreeJustification} for a more detailed comparison with the standard primal-dual approach. 

In Subsection~\ref{SecLPFree}, we formally introduce the LP-free framework. For any given family of valid penalty functions, this framework reduces the competitive ratio analysis of a reduced-bid algorithm to a factor-revealing program over the online and offline usages of a single advertiser. In Subsection~\ref{SecPerfectUtilization}, we use this factor-revealing program as a design tool in the perfect-utilization setting, where the offline usages are fixed at $1$. This allows us to derive the penalty functions used by Algorithm~\ref{alg:LPA} and establish its performance guarantee in this setting. In Subsection~\ref{SecPerfectUtilizationUB}, we establish a matching upper bound, proving the optimality of Algorithm~\ref{alg:LPA} under perfect utilization and completing the proof of Theorem~\ref{Thm2}.

\subsection{LP-Free Certificate}
\label{SecLPFree}

For each day $d\in[D]$ and each query $t\in Q_d$, recall that Algorithm~\ref{alg:LPA} chooses
\[
    i^*(t)
    \in
    \arg\max_{j\in I}
    b_{jt}\left(1-f(x_{jd}(t)\mid R_{jd},d)\right),
\]
and matches query $t$ to $i^*(t)$ if the corresponding reduced bid is positive.

With a slight abuse of notation, let $\OPT_{id}\subseteq Q_d$ denote the set of queries matched to advertiser $i$ by \OPT on day $d$. For each advertiser $i$, we construct nonnegative LP-free certificates $\theta_{id}$ and $\lambda_t$ and constants $\alpha,\beta>0$ satisfying
\begin{align}
    \sum_{d\in[D]}
    \left(
        \theta_{id}
        +
        \sum_{t\in\OPT_{id}}\lambda_t
    \right)
    &\geq
    \alpha B_i
    \sum_{d\in[D]}O_{id},
    &&\forall i\in I,
    \label{LPFreeCertificate1}
    \\
    \sum_{i\in I}\sum_{d\in[D]}
    \left(
        \theta_{id}
        +
        \sum_{t\in\OPT_{id}}\lambda_t
    \right)
    &\leq
    \beta\,\ALG.
    \label{LPFreeCertificate2}
\end{align}

The role of this certificate is to avoid directly comparing the online usage $X_{id}$ with the offline usage $O_{id}$ for each advertiser. Such a comparison can be too restrictive because \ALG and \OPT may assign the same query to different advertisers. Instead, for each advertiser $i$ and day $d$, we track 
\[
\theta_{id}+\sum_{t\in\OPT_{id}}\lambda_t. 
\] 
Later, we lower-bound its normalized value, which we refer to as the \emph{pseudo-reward}.

Motivated by the dual-splitting of \citet{devanur2013randomized}, we construct $\theta_{id}$ and $\lambda_t$ as follows. For each day $d\in[D]$ and each query $t\in Q_d$, the maximum reduced-bid component is assigned to $\lambda_t$, while the penalty contribution associated with the selected advertiser is accumulated in $\theta_{i^*(t)d}$ through integration over the advertiser's realized usage on day $d$. Specifically,
\begin{align}
    \lambda_t
    &:=
    b_{i^*(t)t}
    \left(
        1-
        f\left(
            x_{i^*(t)d}(t)
            \mid R_{i^*(t)d},d
        \right)
    \right),
    \label{LPFreeCandidateLambda}
    \\
    \theta_{id}
    &:=
    B_i\int_0^{X_{id}} f(u\mid R_{id},d)\,du.
    \label{LPFreeCandidateTheta}
\end{align}

Once we find constants $\alpha,\beta>0$ satisfying
\eqref{LPFreeCertificate1} and \eqref{LPFreeCertificate2}, the LP-free
certificate immediately gives a lower bound on the competitive ratio of the
reduced-bid algorithm.

\begin{Lemma}[{\citet{goyal2025asymptotically}}]
\label{LemmaLPFreeAbstract}
Suppose that, for an online algorithm \ALG, there exist nonnegative $\theta_{id}$'s and $\lambda_t$'s satisfying
\eqref{LPFreeCertificate1} and \eqref{LPFreeCertificate2} for some
$\alpha,\beta>0$. Then, \ALG is $\alpha/\beta$-competitive.
\end{Lemma}

Indeed, summing \eqref{LPFreeCertificate1} over all advertisers and combining
it with \eqref{LPFreeCertificate2} gives
\[
    \beta\,\ALG
    \geq
    \alpha \sum_{i\in I}B_i\sum_{d\in[D]}O_{id}
    =
    \alpha\,\OPT.
\]

In particular, in the small-bid regime $\gamma=\max_{i, t}b_{it}/B_i\to0$, the above certificate reduces the analysis to the following factor-revealing program (FRP) over the admissible online and offline usages. The proof is deferred to Appendix~\ref{apd_LPfreeCertificate}.

\begin{restatable}{Lemma}{LPFreeCertificate}\label{LPFreeCertificate}
    Consider a reduced-bid algorithm \ALG with valid penalty functions $f(\cdot\mid r,d)$. Recall that $S_{\ALG}$ and $S_{\OPT}$ denote the sets of admissible online and offline usage vectors, respectively. In the small-bid regime, the competitive ratio of \ALG is lower-bounded by
\begin{align}\label{FRP}
    \Gamma_f
    :=
    \min_{\substack{
        (x_1,\ldots,x_D)\in S_{\ALG},\\
        (o_1,\ldots,o_D)\in S_{\OPT}
    }}
    \frac{
        \sum_{d\in[D]}
        \mathcal{L}_f
        \left(
            x_d,o_d,
            D-\sum_{\ell<d}x_\ell
            \mid d
        \right)
    }{
        \sum_{d\in[D]}o_d
    },
\end{align}
where
\[
    \mathcal{L}_f(x,o,r\mid d)
    :=
    \int_0^x f(u\mid r,d)\,du
    +
    o\bigl(1-f(x\mid r,d)\bigr),
\]
for $r\in[0,D]$, $d\in[D]$, $x\in[0,c(r)]$, and $o\in[0,1]$.
\end{restatable}

We refer to $\mathcal{L}_f(x,o,r\mid d)$ as the \emph{pseudo-reward} associated with online usage $x$, offline usage $o$, and remaining horizon budget $r$ on day $d$. For a fixed advertiser $i$ on day $d$, its first term corresponds to the normalized penalty contribution $\theta_{id}/B_i$, while the second term lower-bounds the normalized reduced-bid contribution $\sum_{t\in\OPT_{id}}\lambda_t/B_i$ in the small-bid regime.

\subsection{Penalty Functions from the Perfect-Utilization Setting}
\label{SecPerfectUtilization}

For any reduced-bid algorithm specified by a fixed family of valid penalty functions, the factor-revealing program in \eqref{FRP} gives a lower bound on its competitive ratio by minimizing over both the online usages $(x_d)_{d\in[D]}$ and the offline usages $(o_d)_{d\in[D]}$. For a fixed remaining budget $r$ and day $d$, the offline usage $o$ directly affects the reduced-bid term $1-f(x\mid r,d)$ in the pseudo-reward, while the penalty term $\int_0^x f(u\mid r,d)\,du$ is unaffected. Therefore, different levels of offline utilization may lead to different families of penalty functions, as we later see in the proof of Theorem~\ref{Thm3}. Rather than directly seeking a single family of penalty functions that balances these competing cases, we first consider a more structured setting.

Specifically, we consider the perfect-utilization setting, where \OPT exhausts every advertiser's nominal daily budget on every day. In this setting, $O_{id}=1$ for every $i\in I$ and $d\in[D]$, so the offline usages are fixed and the factor-revealing program in \eqref{FRP} reduces to an optimization only over the online usages.\footnote{The same restriction is also chosen by \citet{mehta2007adwords}, where it is without loss of generality.} Although this setting does not capture the full difficulty of the general problem in our multi-day setting, as Theorem~\ref{Thm3} shows that partial utilization can lead to a strictly worse competitive ratio, we show that solving this restricted program yields the horizon-aware penalty functions used by Algorithm~\ref{alg:LPA}, together with its performance guarantee in the perfect-utilization setting.

Under perfect utilization, the denominator of the FRP in \eqref{FRP} is simply $D$. Thus, to design the penalty functions that maximize $\Gamma_f$ in \eqref{FRP}, it suffices to analyze the minimum cumulative pseudo-reward over the admissible online usages,
\[
    \min_{(x_1,\ldots,x_D)\in S_{\ALG}}
    \sum_{d\in[D]}
    \mathcal{L}_f
    \left(
        x_d,1,
        D-\sum_{\ell<d}x_\ell
        \mid d
    \right).
\]
Rather than optimizing over all $D$ usages simultaneously, we analyze this program recursively. For example, after fixing the first-day usage $x_1\in[0,\delta]$, we can separate the first-day pseudo-reward from the minimum cumulative pseudo-reward over the remaining days:
\[
    \min_{x_1\in[0,\delta]}
    \left\{
        \mathcal{L}_f(x_1,1,D\mid1)
        +
        \min_{\substack{
            (x_2,\ldots,x_D):\\
            (x_1,\ldots,x_D)\in S_{\ALG}
        }}
        \sum_{d=2}^D
        \mathcal{L}_f
        \left(
            x_d,1,
            D-\sum_{\ell<d}x_\ell
            \mid d
        \right)
    \right\}.
\]
The second term can be viewed as a continuation value determined by the remaining horizon budget $D-x_1$ over the remaining $D-1$ days. More generally, we characterize these continuation values recursively by choosing the current-day penalty so that the overall pseudo-reward is equalized over the current-day usage. We show that continuation values admit a simple interpretation in terms of capped exponential random variables, which also leads to the penalty functions defined in Definition~\ref{DefLPF}.

Before proceeding, recall that
\[
    c(x):=\min\{x,\delta\}.
\]
Let
\[
    S_\ell
    :=
    \sum_{k=1}^{\ell} c(Z_k),
\]
where $Z_1,Z_2,\ldots$ are independent exponential random variables with mean one. With this notation, the function introduced in Definition~\ref{DefLPF} can be written as
\[
    \nu_\ell(x)
    :=
    \Pr(S_\ell\geq x).
\]
For convenience, set $S_0:=0$ and $\nu_0(x):=0$. We further define
\begin{equation}\label{DefV}
    V_\ell(x)
    :=
    \mathbb{E}\!\left[\min\{S_\ell,x\}\right]
    =
    \int_0^x \nu_\ell(z)\,dz,
    \qquad
    V_0(x):=0.
\end{equation}

For each day $d\in[D]$ and remaining normalized horizon budget $r \in [0, D]$ at the beginning of day $d$, define
\[
    W_d(r)
    :=
    \min_{\substack{
        (x_d,\ldots,x_D)\in[0,\delta]^{D-d+1}\\
        \sum_{j=d}^D x_j\leq r
    }}
    \sum_{j=d}^D
    \mathcal{L}_f
    \left(
        x_j,1,
        r-\sum_{m=d}^{j-1}x_m
        \mid j
    \right).
\]
Thus, $W_d(r)$ represents the minimum cumulative pseudo-reward from day $d$ through day $D$, when the normalized horizon budget available at the beginning of day $d$ is $r$.

Our main goal in this subsection is to discuss the key ideas behind the following lemma.
\begin{Lemma}\label{Thm2LBLemma}
For every $d\in[D]$ and $r \in [0, D]$, $W_d(r)=V_{D-d+1}(r)$.
\end{Lemma}
To see why this identity is useful, consider the initial state $d=1$ and $r=D$. It gives
\[
    W_1(D) = 
    \min_{(x_1,\ldots,x_D)\in S_{\ALG}}
    \sum_{d\in[D]}
    \mathcal{L}_f
    \left(
        x_d,1,
        D-\sum_{\ell<d}x_\ell
        \mid d
    \right)
=V_D(D),
\]
which immediately yields the competitive ratio $V_D(D)/D$ through the factor-revealing program in \eqref{FRP} in the perfect-utilization setting.

We prove the identity by backward induction on $d$, as the last day provides a natural starting point. If the remaining horizon budget at the beginning of the last day is $r$, as suggested by the thought experiment, the optimal penalty on the last day is the exponential penalty
\[
    f(x\mid r,D)=f_{c(r)}(x)=e^{x-c(r)}.
\]
Under this penalty, regardless of the last-day usage, the last-day pseudo-reward is always
\[
    V_1(r)=1-e^{-c(r)}.
\]
This gives the desired identity for the base case.

Moving backward, suppose that the minimum pseudo-reward from day $d+1$ onward can be represented by $V_{D-d}(r-x)$ after spending $x$ on day $d$. We then choose the day-$d$ penalty so that the current pseudo-reward together with this continuation value is independent of the current-day usage $x$. We identify the resulting common value as $V_{D-d+1}(r)$, yielding the desired backward recursion.

This argument relies on two main ingredients. First, we choose the penalty function so that the current pseudo-reward together with the continuation value is equalized over every feasible current-day usage, which we refer to as the \emph{equalization property}. This equalization ensures that the same guarantee applies uniformly across all advertisers (perhaps with different online usages), allowing us to focus on a single advertiser. Second, we identify the resulting equalized value through a probabilistic interpretation based on a capped mean-one exponential random variable, which we refer to as the \emph{exponential reinterpretation}. We illustrate these two ingredients in the case $D=\delta=2$ and defer some of the details to Appendix~\ref{apdPerfectUtilizationClaims}.

\begin{exampleblock}
\noindent\textbf{Example 3.}
Fix an advertiser, and suppose that its first-day normalized usage is $x_1\in[0,2]$. The remaining horizon budget on the second day is $r=2-x_1$. On the second day, our choice of the penalty function reduces to
\[
    f^{2, 2}(x_2\mid r,2)=e^{x_2-r},
    \qquad x_2\in[0,r].
\]
Substituting it into the pseudo-reward gives
\begin{align}
    \mathcal{L}_f(x_2,1,r\mid2)
    &=
    \int_0^{x_2}e^{u-r}\,du
    +1-e^{x_2-r}
    \nonumber\\
    &=
    \bigl(e^{x_2-r}-e^{-r}\bigr)
    +\bigl(1-e^{x_2-r}\bigr)
    =
    1-e^{-r}
    =
    V_1(r).
    \label{Example3Eq1}
\end{align}
Note that the dependence on $x_2$ cancels. When $r=1$, this recovers the $1-1/e$ guarantee of \textsc{Balance} in the classical Adwords problem \citep{mehta2007adwords, devanur2013randomized}. More generally, this cancellation shows that the last-day pseudo-reward is equalized across all feasible values of $x_2$, which we refer to as equalization property.

Moreover,
\[
    V_1(r)
    =
    \mathbb{E}[\min\{Z_2,r\}],
    \qquad
    Z_2\sim\operatorname{Exp}(1),
\]
which gives the exponential interpretation of the continuation value.

We next determine the first-day penalty. By \eqref{Example3Eq1},
\[
    W_1(2)
    =
    \min_{x_1\in[0,2]}
    \left\{
        \mathcal{L}_f(x_1,1,2\mid1)
        +
        V_1(2-x_1)
    \right\}.
\]
Following the equalization property, define
\[
    H(x)
    :=
    \mathcal{L}_f(x,1,2\mid1)
    +
    V_1(2-x),
\]
and choose $f^{2, 2}(x\mid2,1)$ so that $H(x)$ is constant over $x\in[0,2]$.
The condition $H'(x)=0$, together with the boundary condition $f^{2, 2}(2\mid2,1)=1$, gives a unique solution of
\[
    f^{2, 2}(x\mid2,1)
    =
    (3-x)e^{x-2},
\]
matching Definition~\ref{DefLPF}. Under this choice, for every $x \in [0, 2]$,
\[
    H(x)
    =
    2-4e^{-2}
    =
    V_2(2) = W_1(2).
\]
For comparison, suppose that we instead use the exponential penalty
\[
    \widetilde f(x)=e^{x-2}
\]
on the first day as well. The corresponding two-day pseudo-reward is
\[
    \widetilde H(x)
    =
    2-e^{-2}-e^{x-2},
\]
which is minimized at $x=2$. Hence, under this exponential first-day penalty, the LP-free analysis certifies only the lower total pseudo-reward of
\[
    \min_{x\in[0,2]}\widetilde H(x)
    =
    1-e^{-2} < 2 - 4e^{-2}.
\]
\end{exampleblock}

Based on these constructions, we formalize the equalization property and the exponential reinterpretation in the next lemma, whose proof is deferred to Appendix~\ref{apdPerfectUtilizationClaims}.

\begin{restatable}{Lemma}{LemmaEqualization}\label{LemmaEqualization}
For each day $d\in[D]$ and remaining horizon budget $r\in[0,D]$, let
\[
    H_{d,r}(x)
    :=
    \mathcal{L}_f(x,1,r\mid d)
    +
    V_{D-d}(r-x),
    \qquad
    x\in[0,c(r)].
\]
Then,
\begin{enumerate}
    \renewcommand{\labelenumi}{(\roman{enumi})}

    \item \textbf{Equalization.}
    The function $H_{d,r}(x)$ is constant over $x\in[0,c(r)]$.

    \item \textbf{Exponential reinterpretation.}
    The constant in part~(i) is $V_{D-d+1}(r)$.
\end{enumerate}
\end{restatable}

Now, we are ready to prove Lemma~\ref{Thm2LBLemma}.

\begin{proof}[Proof of Lemma~\ref{Thm2LBLemma}]
We prove the lemma by backward induction on $d$.

For the base case $d=D$, by the definition of $W_D$ and \eqref{LPF},
\[
    W_D(r)
    =
    \min_{0\leq x\leq c(r)}
    \mathcal{L}_f(x,1,r\mid D) = \min_{0 \leq x \leq c(r)} 1 - e^{-c(r)} = V_1(r).
\]

Now fix $d<D$ and suppose that
\[
    W_{d+1}(r')=V_{D-d}(r')
\]
for every $r'\in[0,D]$. Separating the current-day usage from the remaining
usages gives
\[
    W_d(r)
    =
    \min_{0\leq x\leq c(r)}
    \left\{
        \mathcal{L}_f(x,1,r\mid d)
        +
        W_{d+1}(r-x)
    \right\}.
\]
By the induction hypothesis,
\[
    W_d(r)
    =
    \min_{0\leq x\leq c(r)}
    \left\{
        \mathcal{L}_f(x,1,r\mid d)
        +
        V_{D-d}(r-x)
    \right\}.
\]
By Lemma~\ref{LemmaEqualization}, the expression inside the braces is
$V_{D-d+1}(r)$ for every feasible $x$. Hence,
\[
    W_d(r)=V_{D-d+1}(r).
\]
This completes the backward induction.
\end{proof}

\subsection{Optimality of Algorithm~\ref{alg:LPA} in the Perfect-Utilization Setting}
\label{SecPerfectUtilizationUB}

Recall that Theorem~\ref{Thm2} states that, in the perfect-utilization setting, Algorithm~\ref{alg:LPA} achieves the optimal competitive ratio $V_D(D)/D$. Using the LP-free framework, we have constructed the penalty functions in Algorithm~\ref{alg:LPA} through this restricted setting so that the resulting pseudo-rewards satisfy the equalization property. Then, combining Lemma~\ref{LPFreeCertificate} and Lemma~\ref{Thm2LBLemma} proves that Algorithm~\ref{alg:LPA} is at least $V_D(D)/D$-competitive in the perfect-utilization setting.

It now remains to show that this guarantee is optimal: no online algorithm can achieve a competitive ratio greater than $V_D(D)/D$ under perfect utilization. To prove this upper bound, we fix an arbitrary deterministic fractional online algorithm and adaptively construct an instance on which its competitive ratio is at most $V_D(D)/D$. It is sufficient to consider deterministic fractional online algorithms (and thus use an adaptive adversary), as it is Adwords folklore that any upper bound for deterministic fractional online algorithms also applies to randomized online algorithms. Given a randomized online algorithm, the conditional expected allocation at each arrival defines a deterministic fractional online algorithm that remains feasible and achieves exactly the expected reward of the randomized algorithm. We give a formal reduction in Appendix~\ref{apdDetJustification} and use this reduction in the proofs of Theorem~\ref{Thm2} and, later, Theorem~\ref{Thm3}.

Our upper-bound construction is built from a single-day upper-triangular instance that we apply repeatedly across days. This building block allows \OPT to fully utilize the nominal daily budget of every advertiser, while restricting the normalized allocation of the online algorithm. In this upper-bound argument, an \emph{allocation} refers to the normalized amount fractionally assigned to an advertiser, whereas the resulting \emph{usage} is the normalized amount actually charged after imposing the advertiser's spending limit. Thus, if the allocation is at most $a$ and the effective spending limit is $\kappa$, the usage is at most $\min\{a,\kappa\}$. This distinction is useful because, across days, we vary the effective spending limit according to the advertiser's remaining horizon budget while reusing the same bound for allocations.

In particular, in our construction, after appropriately reordering the advertisers in a single-day instance with $N$ advertisers, each with nominal daily budget $B$, the normalized allocation to advertiser $i$ is upper-bounded, up to a vanishing error, by the harmonic function
\[
    H_N(i):=\sum_{k=N-i+1}^{N}\frac{1}{k},
    \qquad i\in[N].
\]

We formalize this construction in the following lemma. Later, we apply the lemma recursively to groups of advertisers with similar remaining horizon budgets, which allows us to incorporate the horizon spending limit across multiple days.

\begin{Lemma}[Single-Day Upper-Triangular Block]
\label{KVVBuildingBlock}
Fix a normalized spending limit $\kappa\geq0$. For every deterministic fractional online algorithm and positive integers $N$ and $B$, there exists a single-day unit-bid instance with advertisers $I=[N]$, each with nominal daily budget $B$, such that:
\begin{enumerate}
    \renewcommand{\labelenumi}{(\roman{enumi})}

    \item \OPT exhausts the nominal daily budget of every advertiser, and hence
    \[
        \OPT=BN.
    \]

    \item The allocation to advertiser $i$ is at most
    \[
        H_N(i)+\frac{N}{B},
    \]
    and under a spending limit $\kappa$, its usage $X_i$ satisfies
    \[
        X_i
        \leq
        \min\left\{
            H_N(i)+\frac{N}{B},
            \kappa
        \right\}.
    \]
\end{enumerate}
\end{Lemma}
\begin{proof}

Fix an arbitrary deterministic fractional online algorithm \ALG. We construct the instance stage by stage. During each stage, we control the allocation so that no active advertiser receives too much, while at the end of the stage we permanently remove the advertiser receiving the smallest allocation. Repeating this process yields the harmonic upper bound on \ALG, while allowing \OPT to fully utilize the nominal daily budget of every advertiser.

The instance consists of $N$ stages, each containing $B$ unit-bid queries. At stage $t$, each arriving query is compatible with every advertiser in $I(t)$ whose (normalized) allocation during the current stage has not yet exceeded
\[
    \frac{1}{|I(t)|}.
\]
Once an advertiser's stage-$t$ allocation exceeds this threshold, the
adversary makes it temporarily incompatible with the remaining queries of
that stage. At the end of the stage, the adversary permanently removes an advertiser $j(t)\in I(t)$ with the smallest stage-$t$ allocation and sets
\[
    I(t+1):=I(t)\setminus\{j(t)\}.
\]
The temporary incompatibilities are then reset before the next stage. Without loss of generality, we relabel indices of advertisers so that $j(t)$ is $t$, and $I(t) = \{t, t+1, \ldots, N\}$.

\medskip
\noindent
\emph{Offline Optimum.}
Since each stage $t$ contains $B$ queries, advertiser $j(t) = t$ remains compatible with every query in stage $t$. The offline optimum can assign all $B$ queries in stage $t$ to advertiser $t$. Doing so for every $t\in[N]$ assigns exactly $B$ units to each advertiser, and $\OPT = BN$.

\medskip
\noindent
\emph{Online Allocation.}
During stage $t$, due to the threshold, an active advertiser can receive an allocation of at most
\[
    \frac{1}{|I(t)|} + \frac{1}{B} =\frac{1}{N-t+1}+ \frac{1}{B}.
\]
We order the advertisers so that advertiser $i$ remains active through stages $1,\ldots,i$. Therefore, its total allocation to $i$ is at most
\[
    H_N(i)+\frac{N}{B}.
\]
After imposing the normalized spending limit $\kappa$, the usage
is bounded by both the allocation and the spending limit. Thus,
\[
    X_i
    \leq
    \min\left\{
        H_N(i)+\frac{N}{B},
        \kappa
    \right\}.
\]
\end{proof}

This building block is a direct extension of the classical $1-1/e$ upper-bound construction for Adwords \citep{kalyanasundaram2000optimal,mehta2007adwords}. As $N\to\infty$, the harmonic bound $H_N(i)$ for a uniformly chosen advertiser converges in distribution to a mean-one exponential random variable $Z$ (see Appendix~\ref{apdHarmonicLimit}). Hence, under a spending limit $\kappa$, the average usage is asymptotically bounded by
\[
    \mathbb{E}[\min\{Z,\kappa\}]
    =
    1-e^{-\kappa}.
\]
Note that setting $\kappa=1$ recovers the classical $1-1/e$ upper bound. In our setting, allowing $\kappa$ to vary with the remaining horizon budget allows us to reuse the same single-day construction across multiple days.

However, after the first application of the upper-triangular block, we can no longer treat all advertisers identically on subsequent days. In particular, the resulting online usages may differ across advertisers, so their remaining horizon budgets and their effective spending limits are no longer identical. To address this issue, at the beginning of each day, we partition advertisers into \emph{slabs} according to their cumulative online usage. Advertisers in the same slab then have nearly identical remaining horizon budgets. Therefore, we can apply a separate upper-triangular block within each slab on that same day, and the arrivals on the day are obtained by combining these blocks for each slab. To deal with the small variation in remaining horizon budgets within a slab, we relax them to a common upper bound before applying the corresponding single-day upper-triangular block. This relaxation can only help the online algorithm. As the variation of online usages within each slab tends to zero, the error from this relaxation vanishes, yielding the target value $V_D(D)$. We illustrate this recursive construction in the following example and Figure~\ref{fig:KVV_twoday}.

\begin{exampleblock}
    \noindent \textbf{Example 4.} To illustrate the construction, consider $D=\delta=2$. Recall the instance of a single-day upper-triangular block of $I = [N]$ in Lemma~\ref{KVVBuildingBlock}, and assume that its arrivals happen on the first day. Let $y_{i1}\in[0,2]$ denote the normalized usage of advertiser $i$ after the first day. Fix an integer $M$ and set $\eta:=2/M$. We partition advertisers according to their first-day normalized usages:
\[
    I_m
    :=
    \left\{
        i:
        (m-1)\eta<y_{i1}\leq m\eta
    \right\},
    \qquad
    m\in[M],
\]
with zero allocation included in $I_1$. For an advertiser in $I_m$, the actual second-day spending limit lies between
\[
    2-m\eta
    \qquad\text{and}\qquad
    2-(m-1)\eta.
\]
We relax this spending limit to 
\[
    \kappa_m:=2-(m-1)\eta
\]
throughout the slab. This relaxation can only help the online algorithm, and it changes the remaining budget by at most $\eta$. We then apply another single-day upper-triangular block separately within each $I_m$ on the second day.

For example, when $M=3$, based on the first-day allocations, $I$ is divided into
\[
    I_1:\ y_{i1}\in[0,2/3],
    \qquad
    I_2:\ y_{i1}\in(2/3,4/3],
    \qquad
    I_3:\ y_{i1}\in(4/3,2],
\]
and the corresponding relaxed second-day limits are
\[
    \kappa_1=2,
    \qquad
    \kappa_2=\frac{4}{3},
    \qquad
    \kappa_3=\frac{2}{3}.
\]
Thus, after rounding the first-day allocation in each slab down to its lower endpoint, the second-day block produces an allocation based on its relaxed limit. As $M\to\infty$ and $N\to\infty$, we show that the total usage of a randomly chosen advertiser is then bounded by
\[
    V_2(2) = \mathbb{E}\left[\min\left\{
        2,\,
        c(Z_1)+c(Z_2)
    \right\}\right],
\]
where $Z_1$ and $Z_2$ are independent mean-one exponential random variables.

\end{exampleblock}

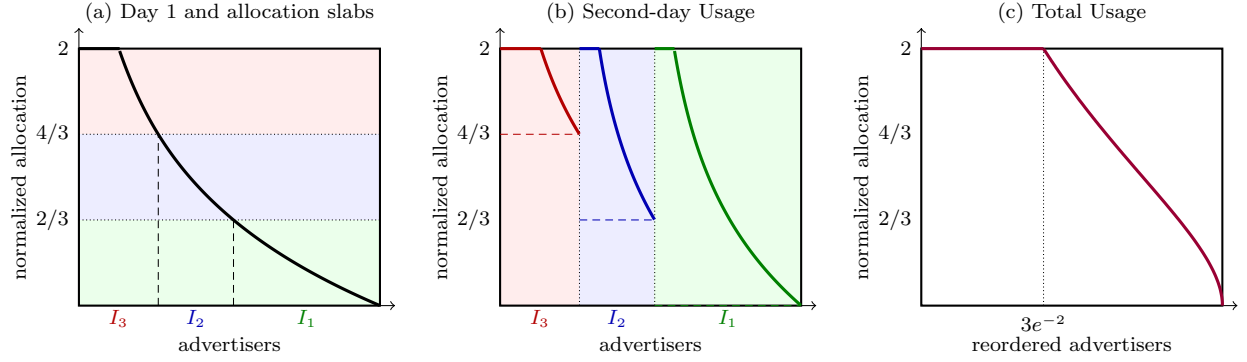
\begin{figure}[t]
    \centering

    \begin{subfigure}[t]{0.325\textwidth}
        \centering
        \resizebox{0.97\linewidth}{!}{%
        \begin{tikzpicture}[x=4.1cm,y=1.75cm,font=\scriptsize]
            \pgfmathsetmacro{\qcap}{exp(-2)}
            \pgfmathsetmacro{\qhigh}{exp(-4/3)}
            \pgfmathsetmacro{\qlow}{exp(-2/3)}
            \pgfmathsetmacro{\midthree}{0.5*\qhigh}
            \pgfmathsetmacro{\midtwo}{0.5*(\qhigh+\qlow)}
            \pgfmathsetmacro{\midone}{0.5*(\qlow+1)}

            \fill[green!8] (0,0) rectangle (1,2/3);
            \fill[blue!7]  (0,2/3) rectangle (1,4/3);
            \fill[red!7]   (0,4/3) rectangle (1,2);

            \draw[->] (0,0) -- (1.05,0);
            \draw[->] (0,0) -- (0,2.16);
            \draw[thick] (0,0) rectangle (1,2);

            \draw[black,very thick] (0,2) -- (\qcap,2);
            \draw[black,very thick,domain=\qcap:1,samples=120,smooth]
                plot (\x,{-ln(\x)});

            \draw[densely dotted] (0,2/3) -- (1,2/3);
            \draw[densely dotted] (0,4/3) -- (1,4/3);

            \draw[densely dashed] (\qhigh,0) -- (\qhigh,4/3);
            \draw[densely dashed] (\qlow,0) -- (\qlow,2/3);

            \node[left] at (0,2/3) {$2/3$};
            \node[left] at (0,4/3) {$4/3$};
            \node[left] at (0,2) {$2$};

            \node[red!70!black] at (\midthree,-0.10) {$I_3$};
            \node[blue!70!black] at (\midtwo,-0.10) {$I_2$};
            \node[green!50!black] at (\midone,-0.10) {$I_1$};

            \node at (0.5,-0.30) {advertisers};
            \node[rotate=90] at (-0.2,1.0) {normalized allocation};
            \node at (0.5,2.28) {(a) Day 1 and allocation slabs};
        \end{tikzpicture}%
        }
    \end{subfigure}
    \hfill
    \begin{subfigure}[t]{0.325\textwidth}
        \centering
        \resizebox{0.97\linewidth}{!}{%
        \begin{tikzpicture}[x=4.1cm,y=1.75cm,font=\scriptsize]
            \pgfmathsetmacro{\aThree}{0}
            \pgfmathsetmacro{\bThree}{exp(-4/3)}
            \pgfmathsetmacro{\aTwo}{\bThree}
            \pgfmathsetmacro{\bTwo}{exp(-2/3)}
            \pgfmathsetmacro{\aOne}{\bTwo}
            \pgfmathsetmacro{\bOne}{1}

            \pgfmathsetmacro{\wThree}{\bThree-\aThree}
            \pgfmathsetmacro{\wTwo}{\bTwo-\aTwo}
            \pgfmathsetmacro{\wOne}{\bOne-\aOne}

            \pgfmathsetmacro{\sThree}
                {\aThree+\wThree*exp(-2/3)}
            \pgfmathsetmacro{\sTwo}
                {\aTwo+\wTwo*exp(-4/3)}
            \pgfmathsetmacro{\sOne}
                {\aOne+\wOne*exp(-2)}

            \pgfmathsetmacro{\midthree}{0.5*(\aThree+\bThree)}
            \pgfmathsetmacro{\midtwo}{0.5*(\aTwo+\bTwo)}
            \pgfmathsetmacro{\midone}{0.5*(\aOne+\bOne)}

            \fill[red!7]   (\aThree,0) rectangle (\bThree,2);
            \fill[blue!7]  (\aTwo,0) rectangle (\bTwo,2);
            \fill[green!8] (\aOne,0) rectangle (\bOne,2);

            \draw[->] (0,0) -- (1.05,0);
            \draw[->] (0,0) -- (0,2.16);
            \draw[thick] (0,0) rectangle (1,2);

            \draw[densely dotted] (\bThree,0) -- (\bThree,2);
            \draw[densely dotted] (\bTwo,0) -- (\bTwo,2);

            \draw[red!70!black,densely dashed]
                (\aThree,4/3) -- (\bThree,4/3);
            \draw[blue!70!black,densely dashed]
                (\aTwo,2/3) -- (\bTwo,2/3);
            \draw[green!50!black,densely dashed]
                (\aOne,0) -- (\bOne,0);

            \draw[red!70!black,very thick]
                (\aThree,2) -- (\sThree,2);
            \draw[red!70!black,very thick,
                  domain=\sThree:\bThree,samples=90,smooth]
                plot
                (\x,{4/3-ln((\x-\aThree)/\wThree)});

            \draw[blue!70!black,very thick]
                (\aTwo,2) -- (\sTwo,2);
            \draw[blue!70!black,very thick,
                  domain=\sTwo:\bTwo,samples=90,smooth]
                plot
                (\x,{2/3-ln((\x-\aTwo)/\wTwo)});

            \draw[green!50!black,very thick]
                (\aOne,2) -- (\sOne,2);
            \draw[green!50!black,very thick,
                  domain=\sOne:\bOne,samples=90,smooth]
                plot
                (\x,{-ln((\x-\aOne)/\wOne)});

            \node[left] at (0,2/3) {$2/3$};
            \node[left] at (0,4/3) {$4/3$};
            \node[left] at (0,2) {$2$};

            \node[red!70!black] at (\midthree,-0.10) {$I_3$};
            \node[blue!70!black] at (\midtwo,-0.10) {$I_2$};
            \node[green!50!black] at (\midone,-0.10) {$I_1$};

            \node at (0.5,-0.30) {advertisers};
            \node[rotate=90] at (-0.2,1.0) {normalized allocation};
            \node at (0.5,2.28) {(b) Second-day Usage};
        \end{tikzpicture}%
        }
    \end{subfigure}
    \hfill
    \begin{subfigure}[t]{0.325\textwidth}
        \centering
        \resizebox{0.97\linewidth}{!}{%
        \begin{tikzpicture}[x=4.1cm,y=1.75cm,font=\scriptsize]
            \pgfmathsetmacro{\sat}{3*exp(-2)}

            \draw[->] (0,0) -- (1.05,0);
            \draw[->] (0,0) -- (0,2.16);
            \draw[thick] (0,0) rectangle (1,2);

            \draw[densely dotted] (\sat,0) -- (\sat,2);
            \node[below] at (\sat,0) {$3e^{-2}$};

            \draw[purple!80!black,very thick]
                (0,2) -- (\sat,2);
            \draw[purple!80!black,very thick,
                  domain=0:2,samples=150,smooth]
                plot ({(1+\x)*exp(-\x)},{\x});

            \node[left] at (0,2/3) {$2/3$};
            \node[left] at (0,4/3) {$4/3$};
            \node[left] at (0,2) {$2$};

            \node at (0.5,-0.30) {reordered advertisers};
            \node[rotate=90] at (-0.2,1.0) {normalized allocation};
            \node at (0.5,2.28) {(c) Total Usage};
        \end{tikzpicture}%
        }
    \end{subfigure}

    \caption{Recursive upper-bound construction for $D=\delta=2$ with
    $M=3$. Panel (a) partitions the set of advertisers based on the first-day usage. Panel (b) relaxes the remaining horizon budget for each slab and applies a single-day upper-triangular block on its second day. Panel (c) reorders advertisers by their total two-day allocation. As the slab becomes arbitrarily thin, the resulting upper-bound profile converges to
    $\mathbb{E}\left[\min\{2,c(Z_1)+c(Z_2)\}\right]$.}
    \label{fig:KVV_twoday}
\end{figure}

Now, we are ready to prove Theorem~\ref{Thm2}.

\TheoremTwo*

\begin{proof}
Recall that combining Lemma~\ref{LPFreeCertificate} and Lemma~\ref{Thm2LBLemma} proves that Algorithm~\ref{alg:LPA} is at least $V_D(D)/D$-competitive in the perfect-utilization setting. It remains to prove the matching upper bound. Specifically, for every $\varepsilon>0$ and every deterministic fractional online algorithm \ALG, we construct an instance in the perfect-utilization setting such that
\[
    \frac{\ALG}{\OPT}
    \leq
    \frac{V_D(D)}{D}+\varepsilon.
\]

The instance consists of $N$ advertisers, each with nominal daily budget $B$. For advertiser $i$, let
\[
    Y_{id}:=\sum_{\ell=1}^d X_{i\ell},
    \qquad
    Y_{i0}:=0,
\]
denote its cumulative normalized usage through day $d \in [D]$. Fix an integer $M$ and let $\eta:=D/M$ and $\ell_m:=(m-1)\eta$. At the beginning of each day $d$, partition the advertisers according to their cumulative usages:
\[
    I_{md}
    :=
    \left\{
        i:
        \ell_m<Y_{i(d-1)}\leq \ell_m+\eta
    \right\},
    \qquad m\in[M],
\]
with zero usage included in $I_{1d}$. 

For each nonempty slab $I_{md}$, we relax the effective spending limit to $\kappa_m:=c(D-\ell_m)$ and apply an independent single-day upper-triangular block from Lemma~\ref{KVVBuildingBlock} on each day $d$.

\medskip
\noindent
\emph{Offline Optimum.} Since every advertiser appears in exactly one slab on every day, \OPT exhausts one nominal daily budget from every advertiser on every day. Therefore, the constructed instance satisfies perfect utilization and $\OPT = NBD$.

\medskip
\noindent
\emph{Online Allocation.}  
Since \[\ALG = B \sum_{i=1}^N Y_{iD},\]
our goal is to upper bound the sum of cumulative usage through day $d$ for each $d \in [D]$ recursively. 

To bound the term $\sum_{i=1}^N Y_{id}$, define
\[
    \Phi_d(y)
    :=
    y+V_{D-d+1}(D-y),
    \qquad
    \Phi_{D+1}(y):=y.
\]

By the definition of $V_d(r)$ in \eqref{DefV}, the functions $\Phi_d$ are non-decreasing and $1$-Lipschitz, and
\begin{equation}\label{eq:UpperBoundPhiRecursion}
    \Phi_d(y)
    =
    \mathbb{E}\left[
        \Phi_{d+1}
        \left(
            y+\min\{Z,c(D-y)\}
        \right)
    \right].
\end{equation}

As the number of advertisers in a slab becomes large, from Lemma~\ref{KVVBuildingBlock}, the harmonic allocation bound of a randomly chosen advertiser is converted into the capped exponential bound $\min\{Z,\kappa_m\}$, where $Z\sim\operatorname{Exp}(1)$. Therefore, up to vanishing errors, the average cumulative usage of advertisers in slab $I_{md}$ after day $d$ is bounded by
\[
    \ell_m+\eta+\min\{Z,\kappa_m\}.
\]

Consequently, applying the upper-triangular block separately within each
slab and using \eqref{eq:UpperBoundPhiRecursion} gives
\begin{equation}\label{eq:UpperBoundOneDayPotential}
    \frac{1}{N}\sum_{i=1}^N
    \Phi_{d+1}(Y_{id})
    \leq
    \frac{1}{N}\sum_{i=1}^N
    \Phi_d(Y_{i,d-1})
    +\varepsilon_{M,N,B},
\end{equation}
where $\varepsilon_{M,N,B}$ vanishes by first taking the slabs sufficiently thin and then choosing $N$ and $B$ sufficiently large. The formal derivation of \eqref{eq:UpperBoundOneDayPotential} is given in Appendix~\ref{apdPerfectUtilizationUB}.

Iterating \eqref{eq:UpperBoundOneDayPotential} over $d=1,\ldots,D$ yields
\[
    \frac{\ALG}{NB}
    =
    \frac{1}{N}\sum_{i=1}^N Y_{iD}
    \leq
    \Phi_1(0)+o(1)
    =
    V_D(D)+o(1).
\]
Since $\OPT=NBD$, for every $\varepsilon>0$, this construction gives a perfect-utilization instance with
\[
    \frac{\ALG}{\OPT}
    \leq
    \frac{V_D(D)}{D}+\varepsilon.
\]
\end{proof}

\section{Analysis in the General Setting}
\label{SecMainResult}

The previous section derives the penalty functions used by Algorithm~\ref{alg:LPA} from the perfect-utilization setting and proves their optimality within this restricted setting. We now return to the general setting, where the offline usage $O_{id}$ may not be $1$. This raises two natural questions: whether the same penalty functions remain effective without perfect utilization, and whether perfect utilization captures the finite-horizon worst case.

In Subsection~\ref{SecGeneralSetting}, we first give a positive answer to the former question in the asymptotic regime through Theorem~\ref{Thm1},  
which states that the competitive ratio of Algorithm~\ref{alg:LPA} converges to $1-\exp(-\delta)$ for every fixed $\delta$ as $D\to\infty$. Thus, the analysis in the perfect-utilization setting identifies the right penalty functions for extracting the power of budget flexibility.

Regarding the latter question, however, we show in Subsection~\ref{SecFiniteGeneral} that the perfect-utilization setting does not characterize the finite-horizon worst case by proving Theorem~\ref{Thm3}, which shows that for $D=\delta=2$, the optimal competitive ratio in the general setting is approximately $0.701$, strictly below $V_2(2)/2\approx0.729$ from Theorem~\ref{Thm2}.  

\subsection{Asymptotic Optimality of Algorithm~\ref{alg:LPA}}
\label{SecGeneralSetting}

We now analyze Algorithm~\ref{alg:LPA} in the general setting. With the penalty functions fixed, the main challenge is to understand how partial offline utilization affects the FRP in \eqref{FRP} when the offline usages $o_d$ may vary across days. In particular, for fixed $\delta$, we aim to identify $\Gamma_D$ such that
\begin{equation}\label{GoalThm1}
    \sum_{d\in[D]}
    \mathcal{L}_f
    \left(
        x_d,o_d,
        D-\sum_{\ell<d}x_\ell
        \mid d
    \right)
    \geq
    \Gamma_D\sum_{d\in[D]}o_d,
\end{equation}
for every admissible usage $(x_1,\ldots,x_D)\in S_{\ALG}$ and $(o_1,\ldots,o_D)\in S_{\OPT}$, with $\Gamma_D\to1-\exp(-\delta)$ as $D\to\infty$.

On the last day, the penalty function reduces to the exponential penalty, so the last-day pseudo-reward can be stated explicitly. However, unlike in the perfect-utilization setting, the resulting value still depends on both the remaining horizon budget, which is determined by $x_1,\ldots,x_{D-1}$, and the offline usage $o_D$. Moving backward therefore does not immediately reduce the problem to a function of the remaining horizon budget alone. Instead of solving this higher-dimensional recursion directly, we exploit the equalization property obtained under perfect utilization to remove the dependence on the online usages, leaving only the effect of the offline usages to be analyzed.

To elaborate in detail, for each day $d$, let
\[
    R_d:=D-\sum_{\ell<d}x_\ell,
\]
and we can rewrite
\[
    \mathcal{L}_f(x_d,o_d,R_d\mid d)
    =
    \mathcal{L}_f(x_d,1,R_d\mid d)
    -
    (1-o_d)\bigl(1-f(x_d\mid R_d,d)\bigr).
\]
Summing over the horizon and using
Lemma~\ref{LemmaEqualization} gives
\begin{align}
    &\sum_{d\in[D]}
    \mathcal{L}_f(x_d,o_d,R_d\mid d)
    \nonumber\\
    &=
    V_D(D)
    -
    \sum_{d\in[D]}
    (1-o_d)\bigl(1-f(x_d\mid R_d,d)\bigr)
    \label{Intermediate2Thm1}\\
    &\geq
    V_D(D)
    -
    \sum_{d\in[D]}
    (1-o_d)\bigl(1-f(0\mid R_d,d)\bigr)
    \nonumber\\
    &=
    V_D(D)
    -
    \sum_{d\in[D]}
    (1-o_d)
    \bigl(V_{D-d+1}(R_d)-V_{D-d}(R_d)\bigr)
    \nonumber\\
    &\geq
    V_D(D)
    -
    \sum_{d\in[D]}
    (1-o_d)
    \bigl(V_{D-d+1}(D)-V_{D-d}(D)\bigr)
    \nonumber\\
    &=
    \sum_{d\in[D]}
    \bigl(V_{D-d+1}(D)-V_{D-d}(D)\bigr) o_d \nonumber\\
    & \geq \bigl(V_{D}(D)-V_{D-1}(D)\bigr)\sum_{d\in[D]}
    o_d \label{IntermediateThm1}.
\end{align}

The first inequality follows from the monotonicity of
the penalty functions. The subsequent equality follows
from the equalization property in
Lemma~\ref{LemmaEqualization}, which gives
\[
    1-f(0\mid r,d)
    =
    V_{D-d+1}(r)-V_{D-d}(r).
\]
The second and the third inequality then follows from the following lemma.

\begin{Lemma}
\label{LemmaMarginalValues}
For every $\ell\geq1$, the function
$V_\ell(r)-V_{\ell-1}(r)$ is non-decreasing in $r$.
Moreover, for every $\ell\geq2$ and $r\geq0$,
\[
    V_\ell(r)-V_{\ell-1}(r)
    \leq
    V_{\ell-1}(r)-V_{\ell-2}(r).
\]
\end{Lemma}
\begin{proof}
Recall that
\[
    S_\ell:=\sum_{k=1}^{\ell}c(Z_k),
    \qquad S_0:=0,
\]
where $Z_1,Z_2,\ldots$ are independent exponential
random variables with mean one. Let
$Y_\ell:=c(Z_\ell)$, which is independent of
$S_{\ell-1}$. By the definition of $V_\ell$,
\begin{align*}
    V_\ell(r)-V_{\ell-1}(r)
    &=
    \mathbb E\left[
        \min\{r,S_{\ell-1}+Y_\ell\}
        -\min\{r,S_{\ell-1}\}
    \right] =
    \mathbb E\left[
        \min\{Y_\ell,(r-S_{\ell-1})^+\}
    \right].
\end{align*}
This expression is non-decreasing in $r$.

Moreover, since $S_{\ell-1}\geq S_{\ell-2}$, for every $\ell\geq2$, we have
\begin{align*}
    V_\ell(r)-V_{\ell-1}(r)
    &=
    \mathbb E\left[
        \min\{Y_\ell,(r-S_{\ell-1})^+\}
    \right] \leq
    \mathbb E\left[
        \min\{Y_\ell,(r-S_{\ell-2})^+\}
    \right] =
    V_{\ell-1}(r)-V_{\ell-2}(r).
\end{align*}
\end{proof}

The preceding analysis establishes the finite-horizon guarantee of Algorithm~\ref{alg:LPA} in the general setting.

\begin{GeneralTheoremRestated}
For every $D\geq1$ and $\delta$,
Algorithm~\ref{alg:LPA} achieves a competitive ratio of at least
\[
    \Gamma_D:=V_D(D)-V_{D-1}(D)
\]
in the small-bid regime.
\end{GeneralTheoremRestated}

\begin{proof}
The result follows directly from \eqref{IntermediateThm1} and Lemma~\ref{LPFreeCertificate}.
\end{proof}

\paragraph{Remark.}
In fact, $\Gamma_D=V_D(D)-V_{D-1}(D)$ is the exact value of the factor-revealing program in \eqref{FRP} for Algorithm~\ref{alg:LPA}. To see this, choose $x_d=0$ for every $d$, $o_1=1$, and $o_d=0$ for all $d\geq2$. Then,
\[
    \frac{
        \sum_{d\in[D]}
        \mathcal L_f(x_d,o_d,R_d\mid d)
    }{
        \sum_{d\in[D]}o_d
    }
    =
    1-f(0\mid D,1)
    =
    V_D(D)-V_{D-1}(D)
    =
    \Gamma_D.
\]

We now show that the finite-horizon guarantee converges to the optimal competitive ratio of the thought experiment, proving Corollary~\ref{CorAsymptotic}.

\CorAsymptotic*

\begin{proof}[Proof of Corollary~\ref{CorAsymptotic}]
Recall that $S_D=S_{D-1}+Y_D$, where
$Y_D:=c(Z_D)=\min\{Z_D,\delta\}$
is independent of $S_{D-1}$ and
$Z_D\sim\operatorname{Exp}(1)$.
By the definition of $V_D$, we have
\[
\begin{aligned}
    \Gamma_D
    &=\mathbb E\left[
        \min\{D,S_D\}-\min\{D,S_{D-1}\}
    \right]\\
    &=\mathbb E\left[
        \min\{Y_D,(D-S_{D-1})^+\}
    \right].
\end{aligned}
\]
Our goal is to show that $(1-e^{-\delta})-\Gamma_D$ vanishes as $D\to\infty$. Since $Y_D\leq\delta$ and $S_D=S_{D-1}+Y_D$,
we have
\[
    0\leq
    Y_D-\min\{Y_D,(D-S_{D-1})^+\}
    \leq
    \delta\,\mathbf{1}\{S_D>D\}.
\]
Taking expectations and using
$\mathbb E[Y_D]=1-e^{-\delta}$ gives
\[
    0\leq (1-e^{-\delta})-\Gamma_D
    \leq\delta\,\Pr(S_D>D).
\]

Since $\mathbb E[S_D]=D(1-e^{-\delta})$
and $0\leq c(Z_d)\leq\delta$,
Hoeffding's inequality gives
\[
    \Pr(S_D>D)
    \leq
    \exp\left(
        -\frac{2e^{-2\delta}}{\delta^2}D
    \right).
\]
Therefore, for every fixed $\delta\geq1$,
\[
    0\leq (1-e^{-\delta})-\Gamma_D
    \leq
    \delta\exp\left(
        -\frac{2e^{-2\delta}}{\delta^2}D
    \right)
    \longrightarrow0
\]
as $D\to\infty$.
Combining this with Theorem~\ref{Thm1} and the upper bound of $1-e^{-\delta}$ from the thought experiment proves the corollary.
\end{proof}

\subsection{Exact Characterization for $D=\delta=2$}
\label{SecFiniteGeneral}

Throughout this subsection, we set $D=\delta=2$. Theorem~\ref{Thm3} gives an exact characterization of the optimal competitive ratio in the general setting, showing that it is strictly smaller than the corresponding value under perfect utilization. In particular, the optimal competitive ratio is $1-\exp(-h^*)\approx0.701$, where $h^*$ is the unique solution to $(1+h^*)\exp(h^*-2)=1$.

The main difficulty is to identify the penalty functions that maximize $\Gamma_f$ of \eqref{FRP}, where the offline usages $o_1$ and $o_2$ are no longer fixed to $1$. Fortunately, on the second day, if the remaining horizon budget at the beginning of day $2$ is $r$, using the exponential penalty $\exp(x-r)$ is still optimal. The first day, however, must balance different possible offline usages. For example, if we knew in advance that $o_1=o_2=1$, the instance would be in the perfect-utilization setting, for which the optimal first-day penalty is
\[
    f^{2,2}(x\mid2,1)=(3-x)e^{x-2}.
\]
On the other hand, if $o_1=1$ and $o_2=0$, then the exponential penalty $\exp(x-2)$ is again optimal. However, Example 1 shows that blindly using the exponential penalty throughout the horizon can perform strictly worse than $1-1/e$. Therefore, the main issue is to find a single first-day penalty that strikes the perfect balance between these two types of trade-offs.
\TheoremThree*
\begin{proof}
    
We first prove the lower bound by solving the resulting two-day factor-revealing program and constructing a first-day penalty that balances these cases. We then prove a matching upper bound by constructing a two-day hard instance for an arbitrary online algorithm. Together, these arguments establish the exact competitive ratio in Theorem~\ref{Thm3}.

\paragraph{Lower bound.}
We construct a horizon-aware reduced-bid algorithm that is $\left(1-\exp(-h^*)\right)$-competitive. Let $f^*(\cdot\mid r,d)$ denote its penalty functions. On the second day, if the remaining normalized horizon budget is $r$, the thought experiment implies that the optimal penalty is
\[
    f^*(x\mid r,2)=f^{2, 2}(x \mid r, 2) = e^{x-r},
    \qquad 0\leq x\leq r,
\] 
since $c(r)=r$ when $D=\delta=2$. It therefore remains to determine the first-day penalty $f^*(\cdot\mid2,1)$.

We directly analyze the factor-revealing program in \eqref{FRP}. Fix an advertiser, and fix the first-day online usage $x\in[0,2]$, so that the second-day remaining horizon budget is $2-x$. Define
\[
    F^*(x)
    :=
    \int_0^x f^*(s\mid2,1)\,ds.
\]
For any $o_2\in[0,1]$, minimizing the second-day pseudo-reward gives
\[
    \min_{0\leq x_2\leq2-x}
    \left\{
        \int_0^{x_2}e^{y-(2-x)}\,dy
        +o_2\bigl(1-e^{x_2-(2-x)}\bigr)
    \right\}
    =
    o_2\bigl(1-e^{x-2}\bigr),
\]
where the minimum is attained at $x_2=0$. Therefore, \eqref{FRP} gives
\begin{equation}\label{NewFRP}
    \Gamma_{f^*}
    =
    \inf_{\substack{
        x\in[0,2],\;
        o_1,o_2\in[0,1]\\
        o_1+o_2>0
    }}
    \frac{
        F^*(x)
        +o_1\bigl(1-f^*(x\mid2,1)\bigr)
        +o_2\bigl(1-e^{x-2}\bigr)
    }{
        o_1+o_2
    }.
\end{equation}

To further simplify \eqref{NewFRP}, set
\[
    \lambda:=\frac{o_1}{o_1+o_2}\in[0,1].
\]
For a fixed $\lambda$, minimizing over the total offline usage
$o_1+o_2$ gives the equivalent formulation
\[
    \Gamma_{f^*}
    =
    \inf_{\substack{x\in[0,2]\\ \lambda\in[0,1]}}
    \left\{
        F^*(x)\max\{\lambda,1-\lambda\}
        +\lambda\bigl(1-f^*(x\mid2,1)\bigr)
        +(1-\lambda)\bigl(1-e^{x-2}\bigr)
    \right\}.
\]
We provide the formal derivation in
Appendix~\ref{Appendix_Lemma41}.

Therefore, to show that the horizon-aware reduced-bid algorithm is
$\gamma$-competitive, it is sufficient to consider the cases
$\lambda<1/2$, $\lambda=1/2$, and $\lambda>1/2$. Equivalently, it suffices
to verify, for every $x\in[0,2]$,
\[
\begin{aligned}
    F^*(x)+1-e^{x-2}
    &\geq \gamma,
    &&\text{corresponding to } (o_1,o_2)=(0,1),\\
    F^*(x)+2-f^*(x\mid2,1)-e^{x-2}
    &\geq 2\gamma,
    &&\text{corresponding to } (o_1,o_2)=(1,1),\\
    F^*(x)+1-f^*(x\mid2,1)
    &\geq \gamma,
    &&\text{corresponding to } (o_1,o_2)=(1,0).
\end{aligned}
\]

The form of $f^*$ is obtained by making the last constraint tight for smaller values of $x$ and the second constraint tight for larger values of $x$. This choice then automatically satisfies the first constraint for every $x$. Solving these conditions and determining the switching point between the corresponding regimes yields
\[
    f^*(x\mid2,1)
    :=
    \begin{cases}
        e^{x-h^*},
        & 0\leq x\leq 2-h^*,\\[1mm]
        (3-x)e^{x-2},
        & 2-h^*<x\leq2.
    \end{cases}
\]

Under this penalty function, all three inequalities are satisfied with
\[
    \gamma
    =
    1-e^{-h^*}
    \approx 0.701.
\]
Therefore, the corresponding horizon-aware reduced-bid algorithm is $\left(1-\exp(-h^*)\right)$-competitive, establishing the lower bound in Theorem~\ref{Thm3}. We refer readers to Appendix~\ref{Appendix_Lemma41} for this direct calculation. Also, see Figure~\ref{fig:compare_penalties_two_day} for an illustration of this penalty function. The optimal first-day penalty function starts from an exponential penalty and switches to the more conservative perfect-utilization penalty once the first-day usage becomes sufficiently large.

\begin{figure}[h]
    \centering
    \begin{tikzpicture}
        \pgfmathsetmacro{\hstar}{1.2079400316}
        \pgfmathsetmacro{\astar}{2-\hstar}

        \begin{axis}[
            width=0.48\textwidth,
            height=0.30\textwidth,
            xmin=0, xmax=2,
            ymin=0.1, ymax=1.02,
            domain=0:2,
            samples=200,
            axis lines=left,
            xlabel={$x$},
            ylabel={penalty},
            xtick={0,\astar,2},
            xticklabels={$0$,$a^*=2-h^*$,$2$},
            ytick={exp(-2),exp(-\hstar),1},
            yticklabels={\scriptsize $e^{-2}$, \scriptsize  $e^{-h^*}$, \scriptsize $1$},
            legend style={
                at={(1.2,0.08)},
                anchor=south east,
                draw=none,
                fill=none,
                font=\small
            },
            clip=false
        ]

        \addplot[
            black,
            thick,
            dotted
        ]
        {exp(x-2)};
        \addlegendentry{$e^{x-2}$}

        \addplot[
            blue,
            thick,
            dashed
        ]
        {(3-x)*exp(x-2)};
        \addlegendentry{$f^{2,2}(x\mid2,1)$}

        \addplot[
            red,
            very thick,
            domain=0:\astar
        ]
        {exp(x-\hstar)};
        \addplot[
            red,
            very thick,
            domain=\astar:2
        ]
        {(3-x)*exp(x-2)};
        \addlegendentry{$f^*(x\mid2,1)$}

        \addplot[
            gray,
            thick,
            dashdotted
        ]
        coordinates {(\astar,0.1) (\astar,1.02)};

        \node[anchor=west] at (axis cs:1.02,0.95)
            {\scriptsize $(3-x)e^{x-2}$};

        \end{axis}
    \end{tikzpicture}
    \caption{Comparison of first-day penalty functions when $D=\delta=2$. For comparison, $\exp(x-2)$ is the exponential penalty that is optimal when there are no arrivals on the second day, while $f^{2,2}(x\mid2,1)=(3-x)\exp(x-2)$ is the penalty obtained from the perfect-utilization setting. The optimal first-day penalty in the general setting, $f^*(x\mid2,1)$, balances these two forces: on $[0, 2 - h^*]$ it lies between the aggressive exponential penalty and the more conservative perfect-utilization penalty, and on $(2-h^*,2]$ it coincides with $f^{2,2}(x\mid2,1)$.}
    \label{fig:compare_penalties_two_day}
\end{figure}
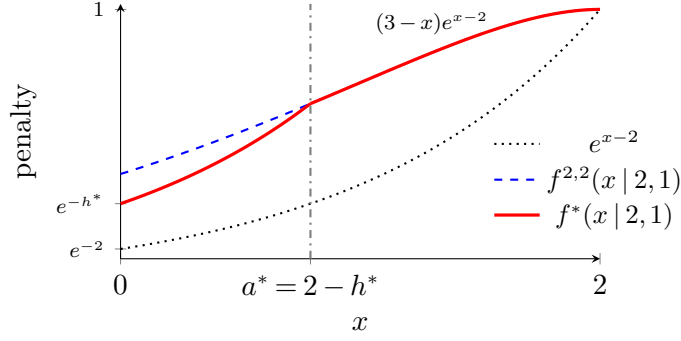

\paragraph{Upper bound.}

We next show that the above guarantee is tight. The construction begins with the upper-bound argument of Theorem~\ref{Thm2}, but differs on the second day. Since perfect utilization is no longer required, after observing the first-day usages, the adversary uses a remaining-budget threshold $h$ and generates second-day arrivals only for advertisers whose remaining horizon budget is at most $h$. Advertisers with remaining budget above $h$ receive no second-day arrivals, so their two-day contribution is determined entirely by the first day. Advertisers with remaining budget at most $h$ receive additional second-day arrivals, but they have already consumed at least $2-h$ of their horizon budgets. This limits how much additional usage \ALG can obtain on the second day, while allowing \OPT to obtain another nominal daily budget. 

Figure~\ref{fig:general_two_day_kvv} illustrates this construction and how the threshold $h=h^*$ determines which advertisers receive second-day upper-triangular blocks. Note that our choice of $h^*$ is the unique minimizer of the competitive ratio illustrated in Figure~\ref{fig:general_two_day_kvv}:
\begin{equation}\label{hRatioEq}
h^*
    :=
    \operatorname*{arg\,min}_{h\in(0,2)}
    \frac{
        \mathbb{E}[Y]
        +
        \mathbb{E}\!\left[
            \mathbf{1}\{Y\geq2-h\}
            \bigl(1-e^{Y-2}\bigr)
        \right]
    }{
        1+\Pr(Y\geq2-h)
    }
    =
    \operatorname*{arg\,min}_{h\in(0,2)}
    \frac{
        1-e^{-2}
        +
        e^{-2}(e^h-1-h)
    }{
        1+e^{h-2}
    },
\end{equation}
with $Y:=\min\{Z,2\}$ and
$Z\sim\operatorname{Exp}(1)$. The unique minimizer satisfies
\[
    (1+h^*)e^{h^*-2}=1,
\]
and the corresponding competitive ratio is $1-e^{-h^*}$.

A potential concern is that \ALG may intentionally limit its first-day usage even when its allocation could be larger, thereby keeping its remaining budget above $h^*$ and avoiding the second-day arrivals.\footnote{For an arbitrary threshold, the ratio obtained from \eqref{hRatioEq} need not be an upper bound against an
online algorithm that deliberately withholds spending.
For example, when $h=0.2$, the corresponding ratio is
\[
    \frac{1+e^{-1.8}-2.2e^{-2}}{1+e^{-1.8}}
    \approx0.7445.
\]
However, an online algorithm can restrict its first-day
spending to up to $1.8 B_i$ for every
advertiser $i$, so that no second-day blocks are generated.
Applying the single-day exponential penalty with this
spending limit guarantees a competitive ratio of at least
$1-e^{-1.8}\approx0.8347$.} However, we show that such withholding cannot improve the resulting competitive ratio. 

\begin{figure}[h]
    \centering

    \def\hstarval{1.207940031569323}

    \begin{subfigure}[t]{0.325\textwidth}
        \centering
        \resizebox{0.97\linewidth}{!}{%
        \begin{tikzpicture}[x=4.1cm,y=1.75cm,font=\scriptsize]
            \pgfmathsetmacro{\hstar}{\hstarval}
            \pgfmathsetmacro{\qcap}{exp(-2)}
            \pgfmathsetmacro{\qth}{exp(\hstar-2)}
            \pgfmathsetmacro{\yth}{2-\hstar}

            \fill[blue!6] (0,\yth) rectangle (\qth,2);

            \draw[->] (0,0) -- (1.05,0);
            \draw[->] (0,0) -- (0,2.16);
            \draw[thick] (0,0) rectangle (1,2);

            \draw[black,very thick] (0,2) -- (\qcap,2);
            \draw[black,very thick,domain=\qcap:1,samples=120,smooth]
                plot (\x,{-ln(\x)});

            \draw[densely dotted] (\qcap,0) -- (\qcap,2);
            \draw[densely dotted] (\qth,0) -- (\qth,\yth);
            \draw[densely dotted] (0,\yth) -- (\qth,\yth);

            \node[left] at (0.13,\yth+0.05) {$2-h^*$};
            \node[left] at (0,2) {$2$};
            \node[below] at (\qcap,0) {$e^{-2}$};
            \node[below] at (\qth,0) {$e^{h^*-2}$};

            \node at (0.66,1.55) {$Y=\min\{X_1,2\}$};

            \node at (0.5,-0.30) {advertisers};
            \node[rotate=90] at (-0.2,1.0) {normalized usage};
            \node at (0.5,2.28) {(a) Day 1 and Possible Usage};
        \end{tikzpicture}%
        }
    \end{subfigure}
    \hfill
    \begin{subfigure}[t]{0.325\textwidth}
        \centering
        \resizebox{0.97\linewidth}{!}{%
        \begin{tikzpicture}[x=4.1cm,y=1.75cm,font=\scriptsize]
            \pgfmathsetmacro{\hstar}{\hstarval}
            \pgfmathsetmacro{\a}{exp(-2)}
            \pgfmathsetmacro{\b}{exp(\hstar-2)}
            \pgfmathsetmacro{\yth}{2-\hstar}

            \pgfmathsetmacro{\sone}{exp(\hstar/4-2)}
            \pgfmathsetmacro{\stwo}{exp(2*\hstar/4-2)}
            \pgfmathsetmacro{\sthree}{exp(3*\hstar/4-2)}

            \fill[blue!6] (0,0) rectangle (\b,2);

            \draw[->] (0,0) -- (1.05,0);
            \draw[->] (0,0) -- (0,2.16);
            \draw[thick] (0,0) rectangle (1,2);

            \foreach \xx in {\a,\sone,\stwo,\sthree,\b} {
                \draw[densely dotted] (\xx,0) -- (\xx,2);
            }

            \draw[densely dotted] (0,\yth) -- (\b,\yth);

            \foreach \L/\R in {
                \a/\sone,
                \sone/\stwo,
                \stwo/\sthree,
                \sthree/\b
            } {
                \pgfmathsetmacro{\ymin}{-ln(\R)}
                \pgfmathsetmacro{\rslab}{2-\ymin}
                \pgfmathsetmacro{\vstar}{exp(-\rslab)}
                \pgfmathsetmacro{\xstar}{\L+(\R-\L)*\vstar}

                \draw[blue,very thick] (\L,2) -- (\xstar,2);
                \draw[blue,very thick,domain=\vstar:1,samples=60,smooth]
                    plot ({\L+(\R-\L)*\x},{\ymin-ln(\x)});
            }

            \node[left] at (0.13,\yth) {$2-h^*$};
            \node[left] at (0,2) {$2$};
            \node[below] at (\b,0) {$e^{h^*-2}$};

            \node at (0.5,-0.30) {advertisers};
            \node[rotate=90] at (-0.2,1.0) {normalized usage};
            \node at (0.5,2.28) {(b) Second-day Usage};
        \end{tikzpicture}%
        }
    \end{subfigure}
    \hfill
    \begin{subfigure}[t]{0.325\textwidth}
    \centering
    \resizebox{0.97\linewidth}{!}{%
    \begin{tikzpicture}[x=4.1cm,y=1.75cm,font=\scriptsize]
        \pgfmathsetmacro{\hstar}{\hstarval}
        \pgfmathsetmacro{\b}{exp(\hstar-2)}
        \pgfmathsetmacro{\sat}{exp(-\hstar)}
        \pgfmathsetmacro{\yth}{2-\hstar}

        \fill[pink!10] (0,0) rectangle (1,1);
        \fill[pink!18] (0,1) rectangle (\b,2);

        \draw[->] (0,0) -- (1.05,0);
        \draw[->] (0,0) -- (0,2.16);
        \draw[thick] (0,0) rectangle (1,2);

        \draw[pink!70!black,thick]
            (0,2) -- (\b,2) -- (\b,1) -- (1,1);
        \node[pink!70!black] at (0.78,1.14) {\scriptsize OPT};

        \draw[red!80!black,very thick] (0,2) -- (\sat,2);

        \draw[red!80!black,very thick,
              domain=\yth:2,samples=120,smooth]
            plot ({(\x+\hstar-1)*exp(-\x)},{\x});

        \draw[red!80!black,very thick,
              domain=0:\yth,samples=120,smooth]
            plot ({exp(-\x)},{\x});

        \draw[densely dotted] (\sat,0) -- (\sat,2);
        \draw[densely dotted] (\b,0) -- (\b,\yth);
        \draw[densely dotted] (0,\yth) -- (\b,\yth);

        \node[left] at (0.13,\yth) {$2-h^*$};
        \node[left] at (0,1) {$1$};
        \node[left] at (0,2) {$2$};
        \node[below] at (\sat-0.04,0) {$e^{-h^*}$};
        \node[below] at (\b+0.04,0) {$e^{h^*-2}$};

        \node at (0.5,-0.30) {reordered advertisers};
        \node[rotate=90] at (-0.2,1.0) {normalized usage};
        \node at (0.5,2.28) {(c) Total Usage};
        \end{tikzpicture}%
        }
    \end{subfigure}

    \caption{Upper-bound construction in the general setting for $D=\delta=2$. Panel (a) shows the first-day usage bound and the threshold $2-h^*$. Panel (b) applies second-day upper-triangular blocks only to advertisers whose remaining horizon budget is at most $h^*$. Panel (c) shows the resulting total usage after reordering advertisers in red, together with the offline usage in pink. \OPT exhausts the nominal daily budget of every advertiser on the first day and of the first $e^{h^*-2}$ fraction of advertisers again on the second day.}
    \label{fig:general_two_day_kvv}
\end{figure}
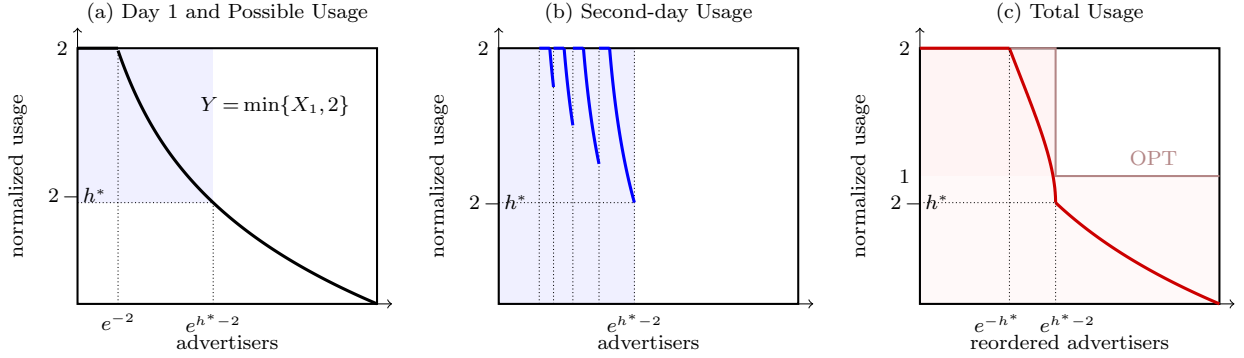

Now, fix an arbitrary deterministic fractional online algorithm \ALG. We construct a $2$-day Adwords instance with $2$-flexibility and advertiser set $I=[N]$, where each advertiser has the same nominal daily budget $B$, such that
\[
    \ALG-(1-e^{-h^*})\OPT
    \leq
    o(NB),
\]
while $\OPT\geq NB$. This implies that, for every $\varepsilon>0$, a
sufficiently large instance satisfies
\[
    \ALG
    \leq
    \bigl(1-e^{-h^*}+\varepsilon\bigr)\OPT.
\]

On the first day, we apply the single-day upper-triangular block from
Lemma~\ref{KVVBuildingBlock} with normalized spending limit $2$. Recall that
$X_{i1}$ denotes the resulting normalized first-day usage of advertiser $i$.

Fix $M\geq1$, set $\eta:=h^*/M$ and
$\ell_m:=2-h^*+(m-1)\eta$, and partition the advertisers with
$X_{i1}>2-h^*$ into slabs
\[
    I_m
    :=
    \left\{
        i:
        \ell_m<X_{i1}\leq\ell_m+\eta
    \right\},
    \qquad
    m\in[M].
\]
On the second day, we apply a separate upper-triangular block to each
nonempty slab $I_m$. Advertisers with $X_{i1}\leq2-h^*$ receive no
second-day queries. As in the proof of Theorem~\ref{Thm2}, within each slab
we relax the remaining horizon budgets to a common upper bound. The resulting
slab and finite-size errors vanish as $M$, $N$, and $B$ become sufficiently
large.

The offline optimum exhausts one nominal daily budget from every advertiser
on the first day and one additional nominal daily budget from every
advertiser receiving a second-day block. Therefore,
\[
    \OPT
    =
    NB+
    B\sum_{i=1}^N
    \mathbf{1}\{X_{i1}>2-h^*\}
    \geq NB.
\]

Consider a slab containing advertisers whose first-day usages are close to $x>2-h^*$. Such an advertiser has remaining normalized horizon budget $2-x$. By Lemma~\ref{KVVBuildingBlock} and Appendix~\ref{apdHarmonicLimit}, the average second-day usage in this slab is asymptotically bounded by
\[
    \mathbb{E}[\min\{Z,2-x\}]
    =
    1-e^{x-2},
    \qquad
    Z\sim\operatorname{Exp}(1).
\]
Consequently, after summing over the slabs,
\[
    \frac{\ALG}{B}
    \leq
    \sum_{i=1}^N
    \left[
        X_{i1}
        +
        \mathbf{1}\{X_{i1}>2-h^*\}
        \bigl(1-e^{X_{i1}-2}\bigr)
    \right]
    +
    o(N).
\]

To compare it with \OPT, for $x\in[0,2]$, define
\[
    \psi(x)
    :=
    x
    +
    \mathbf{1}\{x>2-h^*\}
    \bigl(1-e^{x-2}\bigr)
    -
    (1-e^{-h^*})
    \left(
        1+\mathbf{1}\{x>2-h^*\}
    \right).
\]
The first two terms upper-bound the normalized online value associated with an advertiser whose first-day usage is $x$, while the final term subtracts $1-e^{-h^*}$ times the corresponding offline usage. Formally,
\begin{equation}\label{Eq1thm3}
    \frac{
        \ALG-(1-e^{-h^*})\OPT
    }{B}
    \leq
    \sum_{i=1}^N\psi(X_{i1})
    +
    o(N).
\end{equation}

By construction, a direct calculation shows that $\psi$ is continuous, non-decreasing, and $1$-Lipschitz
on $[0,2]$. Lemma~\ref{KVVBuildingBlock} gives, after relabeling,
\[
    X_{i1}
    \leq
    y_i
    :=
    \min\left\{
        H_N(i)+\frac{N}{B},
        2
    \right\}.
\]
Therefore,
\[
    \sum_{i=1}^N\psi(X_{i1})
    \leq
    \sum_{i=1}^N\psi(y_i).
\]
This inequality also resolves the potential concern that \ALG may intentionally withhold first-day usage below its available allocation to avoid receiving second-day queries: by the monotonicity of $\psi$, decreasing $X_{i1}$ cannot increase its contribution to $\ALG-(1-\exp(-h^*))\OPT$.

Let $U_N$ be uniformly distributed on $[N]$. By
Appendix~\ref{apdHarmonicLimit}, as $N\to\infty$ and $N/B\to0$,
\[
    \frac{1}{N}
    \sum_{i=1}^N\psi(y_i)
    =
    \mathbb{E}\!\left[
        \psi\left(
            \min\left\{
                H_N(U_N)+\frac{N}{B},
                2
            \right\}
        \right)
    \right]
    \longrightarrow
    \mathbb{E}[\psi(Y)],
\]
where
\[
    Y:=\min\{Z,2\},
    \qquad
    Z\sim\operatorname{Exp}(1).
\]

We have
\[
    \mathbb{E}[Y]=1-e^{-2},
    \qquad
    \Pr(Y>2-h^*)=e^{h^*-2},
\]
and
\[
    \mathbb{E}\!\left[
        \mathbf{1}\{Y>2-h^*\}
        \bigl(1-e^{Y-2}\bigr)
    \right]
    =
    \int_{2-h^*}^{2}
        (1-e^{y-2})e^{-y}\,dy
    =
    e^{-2}\bigl(e^{h^*}-1-h^*\bigr).
\]
It follows that
\[
\begin{aligned}
    \mathbb{E}[\psi(Y)]
    &=
    1-e^{-2}
    +
    e^{-2}\bigl(e^{h^*}-1-h^*\bigr)  -
    (1-e^{-h^*})
    \bigl(1+e^{h^*-2}\bigr) \\
    &=
    e^{-h^*}-(1+h^*)e^{-2} \\ &
    =
    0,
\end{aligned}
\]
where the final equality follows from
$(1+h^*)e^{h^*-2}=1$.

Therefore, by \eqref{Eq1thm3},
\[
    \ALG-(1-e^{-h^*})\OPT
    \leq
    o(NB).
\]
Since $\OPT\geq NB$, we obtain
\[
    \frac{\ALG}{\OPT}
    \leq
    1-e^{-h^*}+o(1).
\]
Thus, for every $\varepsilon>0$, choosing $M$, $N$, and $B$ sufficiently
large gives an instance satisfying
\[
    \ALG
    \leq
    \bigl(1-e^{-h^*}+\varepsilon\bigr)\OPT.
\]
\end{proof}

\section{Conclusion}
\label{SecConclusion}

In this paper, we study the power of budget flexibility in the Adwords problem over a $D$-day horizon. Our main result shows that, for every fixed $\delta$, Algorithm~\ref{alg:LPA} achieves the optimal worst-case guarantee asymptotically as $D\to\infty$. This matches the optimal guarantee in the more permissive thought experiment in which the online algorithm receives a fresh budget $\delta B_i$ every day without a horizon spending limit. Therefore, with budget flexibility and the right algorithm, the algorithmic loss due to the horizon spending limit vanishes as the number of days grows. This is in contrast to direct generalizations of \textsc{Balance} based on the exponential penalty, whose worst-case guarantee remains at $1-1/e$ even with budget flexibility.

A key step toward this result is the perfect-utilization setting, in which \OPT exhausts every advertiser's nominal daily budget on every day. Under this restriction, we derive the penalty functions used by Algorithm~\ref{alg:LPA} and simultaneously obtain an exact characterization of the optimal competitive ratio for every $D$ and $\delta$. The lower bound follows from the LP-free framework, while the matching upper bound is established through a multi-day extension of the classical upper-triangular instance. At the same time, the perfect-utilization setting does not always capture the exact finite-horizon worst case of the original model. For $D=\delta=2$, the optimal competitive ratio is approximately $0.701$, strictly below $V_2(2)/2\approx0.729$ under perfect utilization, and the corresponding optimal first-day penalty function blends the horizon-aware penalty with the exponential penalty. Nevertheless, the benefit of this modification vanishes as $D$ increases, as captured by Theorem~\ref{Thm1}.

Several directions remain open. First, it would be interesting to understand the power of budget flexibility beyond the small-bid regime, where each query may consume a non-negligible fraction of an advertiser's budget. Second, characterizing the exact finite-horizon optimum without the perfect-utilization assumption for general $D$ and $\delta$ remains open. The proof of Theorem~\ref{Thm3} suggests a possible structure: different offline-utilization patterns induce different candidate penalty functions, and the optimal penalty may be obtained by balancing these candidates and identifying the appropriate switching points. Extending this approach to general $D$ and $\delta$, however, requires coordinating more candidate functions across different remaining horizon budgets, making the resulting optimization technically challenging.

\paragraph{AI Disclosure}

We used ChatGPT to discuss technical ideas, including the exponential reinterpretation in Lemma~\ref{LemmaEqualization} and solving the system of differential equations from \eqref{NewFRP}. We also used it to discuss proof organization and exposition, and to help polish the manuscript. The authors independently reviewed and verified all AI-assisted content, including the mathematical arguments and references.

\bibliographystyle{informs2014}
\bibliography{ref}

\clearpage
\begin{APPENDICES}
\AppendixNumbering
\section{Omitted Proofs}\label{Missing Proofs}

\subsection{Validity of the Penalty Functions in \eqref{LPF}}
\label{apd_LPFproperties}

Fix $r>0$ and $d\in[D]$, and write $c:=c(r)$. Recall that
\[
    f(x\mid r,d)
    =
    e^{x-c}
    +
    \int_x^c e^{x-z}\nu_{D-d}(r-z)\,dz,
    \qquad x\in[0,c].
\]

We verify that the function $f(\cdot \mid r, d)$ is valid. When $d=D$, the claim is immediate from $f(x\mid r,D)=e^{x-c}$. If $d<D$, we first show that each penalty function takes values in $[\exp(-\delta),1]$ and reaches $1$ at the end of its domain. We then show that it is non-decreasing. Lastly, we derive that the function is $1$-Lipschitz.

For every $x\in(0,c)$,
\begin{equation}\label{eq:LPFDerivative}
    f'(x\mid r,d) = f(x\mid r,d)-\nu_{D-d}(r-x).
\end{equation}

Since the integral term in \eqref{LPF} is nonnegative,
\[
    f(x\mid r,d)
    \geq
    e^{x-c}
    \geq
    e^{-\delta}.
\]
On the other hand, since $\nu_{D-d}\leq1$,
\[
    f(x\mid r,d)
    \leq
    e^{x-c}
    +
    \int_x^c e^{x-z}\,dz
    =
    1.
\]
Thus, $f(x\mid r,d)\in[e^{-\delta},1]\subseteq[0,1]$. Moreover, it is clear that $f(c(r)\mid r,d)=1$.

It remains to show that $f(\cdot\mid r,d)$ is non-decreasing. Let
\[
    S:=\sum_{k=1}^{D-d}c(Z_k),
    \qquad
    Y_x:=\min\{Z_{D-d+1},c-x\}.
\]
By the definition of $\nu_{D-d}$,
\[
\begin{aligned}
     f(x\mid r,d) & \geq  e^{x-c}\nu_{D-d}(r-c)
    +
    \int_x^c e^{x-z}\nu_{D-d}(r-z)\,dz
    \\&=
    \Pr(Y_x+S\geq r-x)\\
    &\geq
    \Pr(S\geq r-x)
    \\&=
    \nu_{D-d}(r-x).
\end{aligned}
\]
By \eqref{eq:LPFDerivative}, $f'(x\mid r,d)\geq0$ for every $x\in(0,c)$, so $f(\cdot\mid r,d)$ is non-decreasing. 

Finally, since $f(x\mid r,d)\leq1$ and $\nu_{D-d}(r-x)\geq0$, \eqref{eq:LPFDerivative} gives
\[
    0\leq f'(x\mid r,d)\leq1,
\]
for every $x\in(0,c)$. Hence, $f(\cdot\mid r,d)$ is $1$-Lipschitz.

\subsection{Proof of Lemma~\ref{LPFreeCertificate}}
\label{apd_LPfreeCertificate}

\LPFreeCertificate*

\begin{proof}
Recall that
\[
    \gamma
    :=
    \max_{i\in I,t}\frac{b_{it}}{B_i}
    \longrightarrow0
\]
in the small-bid regime. We verify \eqref{LPFreeCertificate1} and \eqref{LPFreeCertificate2}
separately with
\[
    \alpha=\Gamma_f,
    \qquad
    \beta=1+O(\gamma).
\]

For a query $t$ matched by \ALG to advertiser $j=i^*(t)$, let $\Delta_t\leq b_{jt}$ denote the amount actually charged to advertiser $j$. The inequality can be strict only when matching the query exhausts the advertiser's effective spending limit. Recall from
\eqref{LPFreeCandidateLambda} and \eqref{LPFreeCandidateTheta} that
$\lambda_t$ is defined using the reduced-bid component of $b_{jt}$, whereas each $\theta_{id}$ represents the integrated penalty contribution over advertiser
$i$'s realized usage on day $d$.
\begin{align*}
    \lambda_t
    &:=
    b_{i^*(t)t}
    \left(
        1-
        f\left(
            x_{i^*(t)d}(t)
            \mid R_{i^*(t)d},d
        \right)
    \right),
    \\
    \theta_{id}
    &:=
    B_i\int_0^{X_{id}} f(u\mid R_{id},d)\,du.
\end{align*}

We first lower bound $\lambda_t$ as follows. Fix an advertiser $i\in I$. Recall that $X_{id}$ and $O_{id}$ denote its
normalized online and offline usages on day $d$, and that
\[
    R_{id}
    =
    D-\sum_{\ell<d}X_{i\ell}.
\]

For every $t\in\OPT_{id}$, the definition of a reduced-bid algorithm gives
\[
    \lambda_t
    \geq
    b_{it}
    \bigl(1-f(x_{id}(t)\mid R_{id},d)\bigr)
    \geq
    b_{it}
    \bigl(1-f(X_{id}\mid R_{id},d)\bigr).
\]
Since the bids of the queries assigned to advertiser $i$ by \OPT sum to at least the amount actually charged by \OPT,
\[
    \frac{1}{B_i}
    \sum_{t\in\OPT_{id}}\lambda_t
    \geq
    O_{id}
    \bigl(1-f(X_{id}\mid R_{id},d)\bigr).
    \tag{A}
\]

By the definition of $\theta_{id}$, combining with (A), for every day $d\in[D]$,
\[
    \frac{1}{B_i}
    \left(
        \theta_{id}
        +
        \sum_{t\in\OPT_{id}}\lambda_t
    \right)
    \geq
    \mathcal{L}_f(X_{id},O_{id},R_{id}\mid d).
\]
Summing over $d\in[D]$ and using the definition of $\Gamma_f$ gives
\[
    \frac{1}{B_i}
    \sum_{d\in[D]}
    \left(
        \theta_{id}
        +
        \sum_{t\in\OPT_{id}}\lambda_t
    \right)
    \geq
    \Gamma_f
    \sum_{d\in[D]}O_{id}.
\]
Therefore, \eqref{LPFreeCertificate1} holds with $\alpha=\Gamma_f$.

It remains to verify our choice of $\beta$. For each query $t$ matched by
\ALG to advertiser $i$ on day $d$, let
\[
    q_t:=\frac{\Delta_t}{B_i}\leq\gamma
\]
denote its normalized actual charge. Note that 
\[
    \theta_{id}
    =
    B_i
    \sum_{\substack{t\text{ matched to }i\\\text{by }\ALG\text{ on day }d}}
    \int_{x_{id}(t)}^{x_{id}(t)+q_t}
        f(u\mid R_{id},d)\,du.
\]
Since $f(\cdot\mid R_{id},d)$ is $C$-Lipschitz, for every such query,
\[
\begin{aligned}
    B_i
    \int_{x_{id}(t)}^{x_{id}(t)+q_t}
        f(u\mid R_{id},d)\,du
    &\leq
    \Delta_t
    f(x_{id}(t)\mid R_{id},d)
    +
    \frac{C B_iq_t^2}{2}\\
    &\leq
    \Delta_t
    f(x_{id}(t)\mid R_{id},d)
    +
    \frac{C\gamma}{2}\Delta_t.
\end{aligned}
\]

Now, since each query is assigned by \OPT to at most one advertiser and
$\lambda_t\geq0$,
\[
    \sum_{i\in I}\sum_{d\in[D]}
        \sum_{t\in\OPT_{id}}\lambda_t
    \leq
    \sum_t \lambda_t
    =
    \sum_{\substack{t\text{ matched}\\\text{by }\ALG}}
    b_{i^*(t)t}
    \left(
        1-
        f\left(
            x_{i^*(t)d}(t)
            \mid R_{i^*(t)d},d
        \right)
    \right),
\]
where the equality follows since $\lambda_t=0$ whenever query $t$ is
discarded by \ALG.

Next, summing the preceding bound on the integrated penalty contribution
over all matched queries gives
\[
\begin{aligned}
    \sum_{i\in I}\sum_{d\in[D]}\theta_{id}
    &\leq
    \sum_{\substack{t\text{ matched}\\\text{by }\ALG}}
    \left[
        \Delta_t
        f\left(
            x_{i^*(t)d}(t)
            \mid R_{i^*(t)d},d
        \right)
        +
        \frac{C\gamma}{2}\Delta_t
    \right].
\end{aligned}
\]
Therefore,
\[
\begin{aligned}
    &\sum_{i\in I}\sum_{d\in[D]}
    \left(
        \theta_{id}
        +
        \sum_{t\in\OPT_{id}}\lambda_t
    \right)\\
    &\leq
    \sum_{\substack{t\text{ matched}\\\text{by }\ALG}}
    \left[
        b_{i^*(t)t}
        \left(
            1-
            f\left(
                x_{i^*(t)d}(t)
                \mid R_{i^*(t)d},d
            \right)
        \right)
        +
        \Delta_t
        f\left(
            x_{i^*(t)d}(t)
            \mid R_{i^*(t)d},d
        \right)
        +
        \frac{C\gamma}{2}\Delta_t
    \right]\\
    &=
    \sum_{\substack{t\text{ matched}\\\text{by }\ALG}}
    \left[
        \Delta_t
        +
        \frac{C\gamma}{2}\Delta_t
        +
        \bigl(b_{i^*(t)t}-\Delta_t\bigr)
        \left(
            1-
            f\left(
                x_{i^*(t)d}(t)
                \mid R_{i^*(t)d},d
            \right)
        \right)
    \right]\\
    &\leq
    \sum_{\substack{t\text{ matched}\\\text{by }\ALG}}
    \left[
        \Delta_t
        +
        \frac{C\gamma}{2}\Delta_t
        +
        \bigl(b_{i^*(t)t}-\Delta_t\bigr)
    \right]\\
    &\leq
    \sum_{\substack{t\text{ matched}\\\text{by }\ALG}}
    \left(
        1+\frac{C\gamma}{2}
    \right)\Delta_t + \sum_{\substack{t\text{ matched}\\\text{by }\ALG}} \bigl(b_{i^*(t)t}-\Delta_t\bigr)\\
    &=
    \left(
        1+\frac{C\gamma}{2}
    \right)\ALG+ \sum_{\substack{t\text{ matched}\\\text{by }\ALG}} \bigl(b_{i^*(t)t}-\Delta_t\bigr).
\end{aligned}
\]
Now, it remains to bound the discrepancy $b_{i^*(t)t}-\Delta_t$. Note that such a discrepancy can be positive only when query $t$ exhausts the effective spending limit of the matched advertiser, and there is at most one such boundary query for each advertiser on each day.

Fix an advertiser $i$, and let $m_i$ denote the number of its boundary queries over the horizon. Every boundary query, except possibly the one that
exhausts the remaining horizon budget, occurs on a day with effective spending
limit at least one and hence contributes at least one normalized unit of
usage. If a boundary query exhausts the remaining horizon budget, then the
total usage of advertiser $i$ is $D\geq m_i$. Therefore, in either case,
\[
    m_i\leq\sum_{d\in[D]}X_{id}.
\]
Since each boundary query can contribute to at most $\gamma B_i$, summing over all advertisers gives
\[
    \sum_{\substack{t\text{ matched}\\\text{by }\ALG}}
    \bigl(b_{i^*(t)t}-\Delta_t\bigr)
    \leq
    \gamma\ALG.
\]

Therefore,
\[
\sum_{i\in I}\sum_{d\in[D]}
    \left(
        \theta_{id}
        +
        \sum_{t\in\OPT_{id}}\lambda_t
    \right) \leq \left(
        1+\left(1+\frac{C}{2}\right)\gamma
    \right)\ALG,
\]
and \eqref{LPFreeCertificate2} holds with
\[
    \beta
    :=
    1+\left(1+\frac{C}{2}\right)\gamma.
\]

Applying Lemma~\ref{LemmaLPFreeAbstract}, the competitive ratio of \ALG is
at least
\[
    \frac{\Gamma_f}{1+\left(1+\frac{C}{2}\right)\gamma}
    =
    \Gamma_f\bigl(1-O(\gamma)\bigr).
\]
Hence, in the small-bid regime, the competitive ratio is lower-bounded by
$\Gamma_f$.
\end{proof}

\subsection{Proof of Lemma~\ref{LemmaEqualization}} \label{apdPerfectUtilizationClaims}
\LemmaEqualization*
\begin{proof}
The case $r=0$ is immediate. Fix $d\in[D]$ and $r\in(0,D]$, and set
\[
    H(x)
    :=
    \mathcal{L}_f(x,1,r\mid d)
    +
    V_{D-d}(r-x).
\]
We show that $H(x) = V_{D-d+1}(r)$ for every $x \in [0, c(r)]$.

By definition,
\[
    \frac{\partial}{\partial x}f(x\mid r,d)
    =
    f(x\mid r,d)-\nu_{D-d}(r-x),
\]
and
\[
    \frac{d}{dx}V_{D-d}(r-x)
    =
    -\nu_{D-d}(r-x)
\]
almost everywhere. Therefore,
\[
    H'(x)
    =
    f(x\mid r,d)
    -
    \frac{\partial}{\partial x}f(x\mid r,d)
    -
    \nu_{D-d}(r-x)
    =
    0.
\]
Hence, $H$ is constant over $[0,c(r)]$. This is the equalization property.

It remains to identify the value of this constant. Let
\[
    Y_r
    :=
    \min\{Z_{D-d+1},c(r)\}.
\]
By construction, $Z_{D-d+1}$ is independent of $S_{D-d}$.

We first formalize the exponential trick. For any continuous function
$h:[0,c(r)]\to\mathbb{R}$ satisfying $h(c(r))=1$,
\begin{equation}\label{eq:CappedExpIdentity}
    \mathbb{E}\left[
        \int_0^{Y_r}h(y)\,dy
        +
        1-h(Y_r)
    \right]
    =
    \mathbb{E}[Y_r]
    =
    1-e^{-c(r)}.
\end{equation}
Indeed,
\[
    \mathbb{E}\left[
        \int_0^{Y_r}h(y)\,dy
    \right]
    =
    \int_0^{c(r)}e^{-y}h(y)\,dy,
\]
while
\[
    \mathbb{E}[1-h(Y_r)]
    =
    \int_0^{c(r)}
        e^{-y}\bigl(1-h(y)\bigr)\,dy
    +
    e^{-c(r)}\bigl(1-h(c(r))\bigr).
\]
Since $h(c(r))=1$, adding the two expressions gives
\[
    \int_0^{c(r)}e^{-y}\,dy
    =
    1-e^{-c(r)}
    =
    \mathbb{E}[Y_r],
\]
which proves \eqref{eq:CappedExpIdentity}.

Since $f(c(r)\mid r,d)=1$ and $H$ is constant, applying
\eqref{eq:CappedExpIdentity} with
$h(\cdot)=f(\cdot\mid r,d)$ gives
\[
    H(x)
    =
    \mathbb{E}[H(Y_r)]
    =
    \mathbb{E}\left[
        Y_r+V_{D-d}(r-Y_r)
    \right] = \mathbb{E}\left[
        \min\{r,S_{D-d+1}\}
    \right]
    =
    V_{D-d+1}(r).
\]
for every $x\in[0,c(r)]$.
This proves the result.
\end{proof}

\subsection{Asymptotic Property of the Harmonic Bound}
\label{apdHarmonicLimit}

Recall that
\[
    H_N(i)
    :=
    \sum_{k=N-i+1}^{N}\frac{1}{k},
    \qquad i\in[N].
\]

By Lemma~\ref{KVVBuildingBlock}, the single-day upper-triangular construction yields the harmonic bound $H_N(i)$ for advertiser $i$. When $U_N$ is chosen uniformly from $[N]$, we show that $H_N(U_N)$ converges in distribution to a mean-one exponential random variable. Consequently, after imposing a spending limit $\kappa$, the average capped harmonic bound converges to $\mathbb{E}[\min\{Z,\kappa\}]$. The following lemma formalizes these observations.

\begin{Lemma}\label{LemmaHarmonicLimit}
Let $U_N$ be chosen uniformly from $[N]$. Then,
\[
    H_N(U_N)
    \xrightarrow{d}
    Z,
    \qquad
    Z\sim\operatorname{Exp}(1),
\]
as $N\to\infty$.
\end{Lemma}

\begin{proof}
Let $U\sim\operatorname{Unif}(0,1)$ and set $U_N:=\lceil NU\rceil$, so that
$U_N$ is uniformly distributed on $[N]$. Note that, as $N \to \infty$, $U_N/N\to U$ almost surely, and
\[
    H_N(U_N)\longrightarrow -\log(1-U)\sim\operatorname{Exp}(1) .
\]
\end{proof}

For any fixed $\kappa\geq0$, the function
$x\mapsto\min\{x,\kappa\}$ is bounded and continuous. Thus, convergence in
distribution implies
\[
    \mathbb{E}[\min\{H_N(U_N),\kappa\}]
    \longrightarrow
    \mathbb{E}[\min\{Z,\kappa\}].
\]
Since $U_N$ is uniform on $[N]$,
\[
    \mathbb{E}[\min\{H_N(U_N),\kappa\}]
    =
    \frac{1}{N}
    \sum_{i=1}^{N}
    \min\{H_N(i),\kappa\}.
\]
Using this property, we can derive the upper bound for the single-day upper-triangular instance from Lemma~\ref{KVVBuildingBlock} as follows. In the construction, the (normalized) usage of advertiser $i$ is at most
\[
    \min\left\{
        H_N(i)+\frac{N}{B},
        \kappa
    \right\},
\]
while $\OPT=BN$. Hence,
\[
    \frac{\ALG}{\OPT}
    \leq
    \frac{1}{N}
    \sum_{i=1}^{N}
    \min\left\{
        H_N(i)+\frac{N}{B},
        \kappa
    \right\}.
\]
For fixed $N$, letting $B\to\infty$ removes the vanishing error $N/B$, and then letting $N\to\infty$ gives
\[
    \limsup \frac{\ALG}{\OPT}
    \leq
    \mathbb{E}[\min\{Z,\kappa\}]
    =
    1-e^{-\kappa}.
\]
Thus, the single-day upper-triangular block yields the asymptotic upper bound $1-e^{-\kappa}$.

\section{Tight Characterization of the Thought Experiment}\label{SecFreshDailyBudget}

Recall the thought experiment introduced in Section~\ref{MainResultsIntro}. The online algorithm \ALG faces no horizon spending limit and instead receives a fresh normalized budget $\delta$ for every advertiser on each day, while \OPT remains subject to its inflexible daily spending limit. Hence, each day can be analyzed independently, and the optimal competitive ratio over $D$ days is the same as that of the corresponding single-day problem. We show that the optimal competitive ratio in this single-day setting is $1-\exp(-\delta)$.

\paragraph{Lower Bound}
We first show that the competitive ratio is at least $1-\exp(-\delta)$ with our choice of penalty function. Recall that $x_{id}(t)$ denotes the normalized amount already spent from advertiser $i$ on day $d$ immediately before the arrival of query $t$. In the single-day setting, we write $x_{i}(t)$ for this quantity. Also, recall from Section~\ref{MainResultsIntro} that
\[
    f_\delta(x):=e^{x-\delta}.
\]

We consider the reduced-bid algorithm that, upon the arrival of query $t$, matches it to an advertiser maximizing
\[
    b_{it}\bigl(1-f_\delta(x_{i}(t))\bigr).
\]
By Lemma~\ref{LPFreeCertificate}, its competitive ratio is lower-bounded by
\begin{equation}\label{FreshBudgetEq1}
    \min_{\substack{x\in[0,\delta]\\ o\in(0,1]}}
    \frac{\int_0^x f_\delta(u)\,du+o\bigl(1-f_\delta(x)\bigr)}{o}
    \geq
    \min_{x\in[0,\delta]}
    \left\{
        e^{x-\delta}-e^{-\delta}+1-e^{x-\delta}
    \right\}
    =
    1-e^{-\delta}.
\end{equation}

\paragraph{Upper Bound}
It remains to show that no (possibly randomized) online algorithm can achieve a competitive ratio greater than $1-\exp(-\delta)$ in the single-day thought experiment. As discussed in Appendix~\ref{apdDetJustification}, it is sufficient to consider an arbitrary deterministic fractional online algorithm.

Apply Lemma~\ref{KVVBuildingBlock} with $\kappa=\delta$. The offline optimum
exhausts the nominal daily budget of every advertiser and obtains
\[
    \OPT=BN.
\]
For the online algorithm, the normalized usage of advertiser $i$ is at most
\[
    \min\left\{
        H_N(i)+\frac{N}{B},
        \delta
    \right\}.
\]
Therefore,
\[
    \ALG
    \leq
    B\sum_{i=1}^{N}
    \min\left\{
        H_N(i)+\frac{N}{B},
        \delta
    \right\}.
\]
For every fixed $N$, sending $B\to\infty$ gives
\[
    \limsup_{B\to\infty}
    \frac{\ALG}{\OPT}
    \leq
    \frac{1}{N}
    \sum_{i=1}^{N}
    c\bigl(H_N(i)\bigr) = \mathbb{E}\!\left[c\bigl(H_N(U_N)\bigr)\right].
\]
By Lemma~\ref{LemmaHarmonicLimit}, since $c(\cdot)$ is bounded and continuous,
\[
    \frac{1}{N}
    \sum_{i=1}^{N}
    c\bigl(H_N(i)\bigr)
    \longrightarrow
    \mathbb{E}[c(Z)]
    =
    1-e^{-\delta},
\]
as $N \to \infty$. Therefore,
\[
    \frac{\ALG}{\OPT}
    \leq
    1-e^{-\delta}+o(1).
\]
Together with the lower bound, this proves that the optimal competitive ratio in the single-day thought experiment is exactly $1-\exp(-\delta)$. Since our main setting further restricts \ALG by imposing the horizon spending limit, the same value also serves as an upper bound in our setting for every $D$ and $\delta$.

\section{A Limitation of Direct Generalizations of \textsc{Balance}}
\label{UBclassicBALANCE}

We extend Example~1 and show that a broad class of direct exponential
generalizations of \textsc{Balance} cannot improve upon the classical
$1-1/e$ guarantee.

Recall that the classical \textsc{Balance} sets its penalty function as $\exp(x-1)$, where $x\in[0,1]$ denotes the normalized amount already spent. If an advertiser can instead spend up to $r>0$ normalized units on the current day, its natural extension is
\[
    f_{\beta,r}(x)
    :=
    \exp\left(
        \beta\left(\frac{x}{r}-1\right)
    \right),
    \qquad
    x\in[0,r],
\]
where $\beta\geq1$ controls the steepness of the penalty. Note that, this is a valid penalty function since it is non-decreasing and satisfies $f_{\beta,r}(r)=1$. Setting $r=\beta=1$ recovers the classical \textsc{Balance} in a single day, although we limit \OPT to be inflexible. Moreover, when the remaining budget is sufficient, i.e., $r \geq \delta$, setting $\beta=\delta$ gives
\[
    f_{\delta, r}(x)=f_{\delta}(x) = e^{x-\delta},
\]
which is the exponential penalty used in the thought experiment.

Recall that $R_{id}$ denotes the remaining horizon budget of advertiser $i$ at the beginning of day $d$. We allow the algorithm to choose a possibly smaller spending limit $s_{id}\in[0,c(R_{id})]$ for each advertiser $i$ at the beginning of day $d$, based on the history before that day. It then applies $f_{\beta,s_{id}}$ on $[0,s_{id}]$. If $s_{id}=0$, the advertiser is inactive on that day. Therefore, the result below also allows the algorithm to intentionally restrict its daily spending below the available limit.

\begin{Lemma}\label{LemmaGeneralizedBalanceUpperBound}
Fix $D\geq2$, $\delta\in[1,D]$, and $\beta\geq1$. Any deterministic
reduced-bid algorithm using the (scaled) exponential penalty
$f_{\beta,s_{id}}$ on every day has competitive ratio at most
\[
    \rho(\beta)
    :=
    \left(
        1+
        \frac{1}{1-e^{-\beta}}
        -
        \frac{1}{\beta}
    \right)^{-1}
    \leq
    1-\frac{1}{e}
\]
in the small-bid regime.
\end{Lemma}

\begin{proof}
Fix an arbitrarily small $\epsilon\in(0,1)$ and a sufficiently large $B$. Consider two advertisers with nominal daily budgets
\[
    B_1:=B,
    \qquad
    B_2:=\delta B.
\]
All bids in this instance are at most one, and the instance lies in the small-bid regime as $B\to\infty$.

We construct the arrivals adaptively over the $D$ days. At the beginning of day $d$, let
\[
    \delta_d:=s_{1d}\in[0,c(R_{1d})]
\]
be the spending limit selected by the algorithm for advertiser $1$.  Suppose first that $\delta_d>0$, and let
$n_d:=\lfloor\delta_dB\rfloor$. The first $n_d$ queries on day $d$ are
compatible with both advertisers. Advertiser $1$ bids one, while advertiser
$2$ bids
\[
    b_{2,d, t}
    :=
    \frac{
        (1-\epsilon)
        \left[
            1-
            \exp\left(
                \beta\left(
                    \frac{t-1}{\delta_dB}-1
                \right)
            \right)
        \right]
    }{
        1-e^{-\beta}
    },
    \qquad
    t\in[n_d].
\]
After these shared queries, $B$ additional unit-bid queries arrive, each compatible only with advertiser $1$. If $\delta_d=0$, there are no shared queries and release only these $B$ advertiser-$1$-only queries.

\medskip
\noindent
\emph{Offline Optimum.} \OPT assigns the $B$ advertiser-$1$-only queries on every day to advertiser $1$ and all shared queries to advertiser $2$. Note that this allocation is feasible under its daily limit, as advertiser $1$ spends exactly $B$ per day, while the total bid assigned to advertiser $2$ on day $d$ is at most $\delta_dB\leq\delta B=B_2$. Therefore,
\[
    \OPT
    \geq
    DB+\sum_{d \in [D]} \sum_{t \in [n_d]}b_{2,d, t}.
\]

For every $\delta_d>0$, as $B \to \infty$,
\[
\begin{aligned}
    \frac{1}{B}\sum_{t=1}^{n_d}b_{2,d,t}
    \approx
    \frac{(1-\epsilon)\delta_d}{1-e^{-\beta}}
    \int_0^1
        \left(1-e^{\beta(u-1)}\right)\,du
    =
    (1-\epsilon)\delta_d
    \left(
        \frac{1}{1-e^{-\beta}}
        -
        \frac{1}{\beta}
    \right).
\end{aligned}
\]

Let
\[
    S:=\sum_{d \in [D]}\delta_d.
\]

Then, as $B \to \infty$,
\[
    \frac{\OPT}{B}
    \geq
    \left[
        D+
        (1-\epsilon)S
        \left(
            \frac{1}{1-e^{-\beta}}
            -
            \frac{1}{\beta}
        \right)
    \right].
\]

\medskip
\noindent
\emph{Online Algorithm.} Suppose that the first $t-1$ shared queries on day $d$ have been assigned to advertiser $1$. Its normalized current-day usage is then $(t-1)/B$, so its reduced bid on query $t$ is
\[
    1-
    \exp\left(
        \beta\left(
            \frac{t-1}{\delta_dB}-1
        \right)
    \right).
\]
Advertiser $2$ has not been matched to any queries. If it is active, its reduced bid is
\[
\begin{aligned}
    b_{2,d,t}(1-e^{-\beta})
    &=
    (1-\epsilon)
    \left[
        1-
        \exp\left(
            \beta\left(
                \frac{t-1}{\delta_dB}-1
            \right)
        \right)
    \right],
\end{aligned}
\]
which is strictly smaller. If advertiser $2$ is inactive, the conclusion is immediate. Therefore, by induction, the algorithm assigns every shared query to advertiser $1$.

Consequently, the algorithm spends up to its effective limit $\delta_d$ from advertiser $1$ on day $d$ and obtains no additional value from the remaining advertiser-$1$-only queries.
Since these are actual daily spending limits subject to advertiser $1$'s
normalized horizon budget, $S\leq D$, and
\[
    \frac{\ALG}{B} \to S,
\]
as $B\to \infty$.

Therefore,
\[
    \limsup_{B\to\infty}\frac{\ALG}{\OPT}
    \leq
    \frac{S}{
        D+
        (1-\epsilon)S
        \left(
            \frac{1}{1-e^{-\beta}}
            -
            \frac{1}{\beta}
        \right)
    } \leq
    \left[
        1+
        (1-\epsilon)
        \left(
            \frac{1}{1-e^{-\beta}}
            -
            \frac{1}{\beta}
        \right)
    \right]^{-1}.
\]
Letting $\epsilon\to0$ proves that the competitive ratio is at most
$\rho(\beta)$.

It remains to compare this bound with $1-1/e$. Let
\[
    G(\beta)
    :=
    \frac{1}{1-e^{-\beta}}
    -
    \frac{1}{\beta}.
\]
For every $\beta>0$, by direct calculation,
\[
    G'(\beta)
    =
    \frac{1}{\beta^2}
    -
    \frac{e^{-\beta}}{(1-e^{-\beta})^2}
    >
    0.
\]
Thus, $G(\beta)$ is strictly increasing and $\rho(\beta)=1/(1+G(\beta))$ is strictly decreasing. At $\beta=1$,
\[
    \rho(1)
    =
    \left(
        1+\frac{1}{e-1}
    \right)^{-1}
    =
    1-\frac{1}{e}.
\]
Therefore, $\rho(\beta)\leq1-1/e$ for every $\beta\geq1$, with strict
inequality for $\beta>1$.
\end{proof}

Note that, when $\beta=\delta=2$, the bound becomes
\[
    \rho(2)
    =
    \frac{2(1-e^{-2})}{3-e^{-2}}
    \approx0.604,
\]
recovering Example~1.

\section{From Primal-Dual to LP-Free Analysis}
\label{apdLPFreeJustification}

In the small-bid regime, it is folklore in Adwords to consider the natural fractional relaxation of the offline optimum. Let $y_{it}\in[0,1]$ denote the fraction of query $t$ assigned to advertiser $i$. In particular, $y_{it}=1$ corresponds to assigning query $t$ entirely to advertiser $i$. In our setting, consider
\begin{align}
    \max\quad
    & \sum_{d\in[D]}\sum_{t\in Q_d}\sum_{i\in I}
        b_{it}y_{it}
    \nonumber\\
    \text{s.t.}\quad
    & \sum_{i\in I}y_{it}\leq1,
        && \forall d\in[D],\ t\in Q_d,
    \label{LPPrimalQuery}\\
    & \sum_{t\in Q_d}b_{it}y_{it}\leq B_i,
        && \forall i\in I,\ d\in[D],
    \label{LPPrimalBudget}\\
    & y_{it}\geq0,
        && \forall i\in I,\ d\in[D],\ t\in Q_d.
    \nonumber
\end{align}
Note that the horizon spending limit is redundant for \OPT, since its
inflexible daily spending limits already imply it. Also, the primal relaxation above can equivalently be viewed as a classical Adwords instance with $D$ copies of each advertiser, one for each day, where queries in $Q_d$ are adjacent only to the corresponding day-$d$ copies. 

Let $\lambda_t$ be the dual variable associated with
\eqref{LPPrimalQuery}, and let $\theta_{id}/B_i$ be the dual variable associated with \eqref{LPPrimalBudget}. The resulting dual program is
\begin{align}
    \min\quad
    & \sum_{d\in[D]}\sum_{t\in Q_d}\lambda_t
      +\sum_{i\in I}\sum_{d\in[D]}\theta_{id}
    \nonumber\\
    \text{s.t.}\quad
    & \frac{b_{it}}{B_i}\theta_{id}+\lambda_t
      \geq b_{it},
        && \forall i\in I,\ d\in[D],\ t\in Q_d,
    \label{LPDualEdge}\\
    & \theta_{id}\geq0,
        && \forall i\in I,\ d\in[D],
    \nonumber\\
    & \lambda_t\geq0,
        && \forall d\in[D],\ t\in Q_d.
    \nonumber
\end{align}

A classical primal-dual analysis of Adwords \citep{buchbinder2007online,devanur2013randomized} can be viewed as constructing nonnegative dual variables, together with a factor $\alpha\geq0$, such that
\[
    \frac{b_{it}}{B_i}\theta_{id}+\lambda_t
    \geq
    \alpha b_{it}
\]
holds for every advertiser-query pair $(i,t)$. Directly applying this pairwise analysis in our setting recovers only the $1-1/e$ guarantee.

To exploit the flexibility and improve beyond this guarantee, one possible approach is to allow the local guarantee associated with each advertiser-day pair to vary across days. In particular, suppose that for advertiser $i$ on day $d$, we obtain a factor $\alpha_{id}\geq0$ such that
\[
    \frac{b_{it}}{B_i}\theta_{id}+\lambda_t
    \geq
    \alpha_{id}b_{it},
    \qquad
    t\in\OPT_{id},
\]
where $\OPT_{id}\subseteq Q_d$ denotes the set of queries matched to advertiser $i$ by \OPT on day $d$. Summing these inequalities over $t\in\OPT_{id}$ gives
\[
    O_{id}\theta_{id}
    +
    \sum_{t\in\OPT_{id}}\lambda_t
    \geq
    \alpha_{id}B_iO_{id}.
\]
Since $O_{id}\leq1$, this further implies
\[
    \theta_{id}
    +
    \sum_{t\in\OPT_{id}}\lambda_t
    \geq
    \alpha_{id}B_iO_{id}.
\]

Directly characterizing the local factors $\alpha_{id}$, however, is less transparent in our setting. This is because the penalty function on day $d$, as well as its domain, depends on the realized remaining horizon budget $R_{id}$, which itself depends on the online usages on earlier days. A natural alternative is therefore to bypass the individual local factors and directly seek an amortized certificate of the form
\[
    \sum_{d\in[D]}
    \left(
        \theta_{id}
        +
        \sum_{t\in\OPT_{id}}\lambda_t
    \right)
    \geq
    \alpha^* B_i
    \sum_{d\in[D]}O_{id},
\]
where
\[
    \alpha^* := \frac{
        \sum_{d\in[D]}\alpha_{id}O_{id}
    }{
        \sum_{d\in[D]}O_{id}
    }.
\]

This amortization across days for each advertiser leads directly to the LP-free framework of \citet{goyal2025asymptotically} used in Section~\ref{SecLPFree}.

\section{Why deterministic fractional algorithms are sufficient for upper bounds}
\label{apdDetJustification}

In a deterministic fractional online algorithm, an arriving query $t$ may be fractionally assigned among advertisers, where $y_{it}\in[0,1]$ denotes the fraction assigned to advertiser $i$. The resulting allocations satisfy
\begin{align}
    \text{maximize} \quad
    & \sum_{i\in I}\sum_{d\in[D]}\sum_{t\in Q_d} b_{it} y_{it}
    \nonumber\\
    \text{subject to} \quad
    & \sum_{t\in Q_d} b_{it} y_{it}
    \leq \delta B_i,
    && \forall i\in I,\ d\in[D],
    \label{LPUpperDailyApp}\\
    & \sum_{d\in[D]}\sum_{t\in Q_d} b_{it} y_{it}
    \leq DB_i,
    && \forall i\in I,
    \label{LPUpperTotalApp}\\
    & \sum_{i\in I} y_{it}\leq1,
    && \forall d\in[D],\ t\in Q_d,
    \label{LPUpperQueryApp}\\
    & 0\leq y_{it}\leq1,
    && \forall i\in I,\ d\in[D],\ t\in Q_d.
    \nonumber
\end{align}
The constraints in \eqref{LPUpperDailyApp} and \eqref{LPUpperTotalApp} enforce the daily and horizon spending limits, respectively, while \eqref{LPUpperQueryApp} ensures that the total fraction assigned from each query is at most one.

To justify the reduction, consider an arbitrary randomized online algorithm. For each arriving query $t$, let $\mathcal{F}_t$ denote the filtration observed up to the arrival of query $t$. Define
\[
    \bar y_{it}
    :=
    \mathbb{E}\!\left[y_{it}\mid\mathcal{F}_t\right].
\]
Now, consider the deterministic fractional online algorithm that assigns query $t$ according to the fractions $(\bar y_{it})_{i\in I}$. Since every realization of the randomized algorithm
satisfies \eqref{LPUpperDailyApp}--\eqref{LPUpperQueryApp}, so does the resulting fractional algorithm, and by linearity of expectation its reward is exactly the expected reward of the randomized algorithm. Hence, any upper bound for deterministic fractional online algorithms also applies to randomized online algorithms.

\section{Additional Details for the Proof of Theorem~\ref{Thm2}}
\label{apdPerfectUtilizationUB}

In this section, we provide the formal derivation of
\eqref{eq:UpperBoundOneDayPotential} used in the proof of
Theorem~\ref{Thm2}. Recall that
\[
    \Phi_d(y)
    :=
    y+V_{D-d+1}(D-y),
    \qquad
    \Phi_{D+1}(y):=y.
\]

We first show that $\Phi_d$ is non-decreasing and $1$-Lipschitz. Indeed, since $V_k'(r)=\nu_k(r)\in[0,1]$ almost everywhere,
\[
    \Phi_d'(y)
    =
    1-\nu_{D-d+1}(D-y)
    \in[0,1].
\]

The definition of $V_k$ also gives the recursion
\begin{equation}\label{eq:PhiRecursionAppendix}
    \Phi_d(y)
    =
    \mathbb{E}\left[
        \Phi_{d+1}
        \left(
            y+\min\{Z,c(D-y)\}
        \right)
    \right],
    \qquad
    Z\sim\operatorname{Exp}(1).
\end{equation}
Indeed, writing $r=D-y$ and
$Y=\min\{Z,c(r)\}$, we have
\[
    V_{D-d+1}(r)
    =
    \mathbb{E}\left[
        Y+V_{D-d}(r-Y)
    \right],
\]
which is equivalent to \eqref{eq:PhiRecursionAppendix}.

We now prove a useful lemma for analyzing the hard instance.

\begin{Lemma}\label{LemmaOneDayPotentialUB}
Fix $M\geq1$ and let $\eta:=D/M$ and
$\ell_m:=(m-1)\eta$. For every $\xi>0$, there exists an integer
$K=K(M,\xi)$ such that, for every day $d\in[D]$,
\[
    \frac{1}{N}\sum_{i=1}^N
    \Phi_{d+1}(Y_{id})
    \leq
    \frac{1}{N}\sum_{i=1}^N
    \Phi_d(Y_{i,d-1})
    +
    \eta+\xi+\frac{N}{B}+\frac{DMK}{N}.
\]
\end{Lemma}

Consequently, the error on the right-hand side can be made arbitrarily small by first choosing $M$ sufficiently large and $\xi$ sufficiently small, then choosing $N$ sufficiently large, and finally choosing $B$ sufficiently large relative to $N$.

Applying Lemma~\ref{LemmaOneDayPotentialUB} for
$d=1,\ldots,D$ gives
\[
    \frac{1}{N}\sum_{i=1}^N Y_{iD}
    \leq
    \Phi_1(0)
    +
    D
    \left(
        \eta+\xi+\frac{N}{B}+\frac{DMK}{N}
    \right).
\]
Since $\Phi_1(0)=V_D(D)$, $\ALG=B\sum_iY_{iD}$, and
$\OPT=NBD$, this yields
\[
    \frac{\ALG}{\OPT}
    \leq
    \frac{V_D(D)}{D}
    +
    \eta+\xi+\frac{N}{B}+\frac{DMK}{N}.
\]
Choosing the parameters based on the order described above gives the desired $V_D(D)/D+\varepsilon$ upper bound for every $\varepsilon >0$. We conclude this subsection with the proof of Lemma~\ref{LemmaOneDayPotentialUB}.

\begin{proof}[Proof of Lemma~\ref{LemmaOneDayPotentialUB}]
Fix a day $d\in[D]$. Recall that
\[
    I_{md}
    =
    \left\{
        i:
        \ell_m<Y_{i,d-1}\leq\ell_m+\eta
    \right\},
\]
with zero usage included in $I_{1d}$, and let
$n_{md}:=|I_{md}|$. For advertisers in $I_{md}$, we use the relaxed
spending limit
\[
    \kappa_m:=c(D-\ell_m).
\]
Since $Y_{i,d-1}\geq\ell_m$, the actual effective spending limit satisfies
\[
    c(D-Y_{i,d-1})\leq \kappa_m,
\]
so this relaxation can only increase the possible usage of the online
algorithm.

Consider first a slab $I_{md}$ of size $n_{md}$. Relabel the advertisers
inside the slab according to their removal order in the single-day
upper-triangular block, and let
$\rho_{id}\in[n_{md}]$ denote the index of advertiser $i$. By
Lemma~\ref{KVVBuildingBlock},
\[
    X_{id}
    \leq
    \min\left\{
        H_{n_{md}}(\rho_{id})+\frac{n_{md}}{B},
        \kappa_m
    \right\}.
\]
Define
\[
    \widehat X_{id}
    :=
    \min\left\{
        H_{n_{md}}(\rho_{id}),
        \kappa_m
    \right\}.
\]
Then
$X_{id}\leq\widehat X_{id}+n_{md}/B$. Since
$Y_{i(d-1)}\leq\ell_m+\eta$,
\[
    Y_{id}
    \leq
    \ell_m+\widehat X_{id}
    +\eta+\frac{n_{md}}{B}.
\]
Using that $\Phi_{d+1}$ is non-decreasing and $1$-Lipschitz,
\begin{equation}\label{eq:SlabLipBound}
    \Phi_{d+1}(Y_{id})
    \leq
    \Phi_{d+1}
    \left(
        \ell_m+\widehat X_{id}
    \right)
    +
    \eta+\frac{n_{md}}{B}.
\end{equation}

Assume that the slab $n_{md} = |I_{md}|$ is sufficiently large. Since the indices $(\rho_{id})_{i\in I_{md}}$ form a permutation of $[n_{md}]$,
\[
    \frac{1}{n_{md}}
    \sum_{i\in I_{md}}
    \Phi_{d+1}
    \left(
        \ell_m+\widehat X_{id}
    \right)
    =
    \frac{1}{n_{md}}
    \sum_{j \in [n_{md}]}
    \Phi_{d+1}
    \left(
        \ell_m+
        \min\{H_{n_{md}}(j),\kappa_m\}
    \right).
\]
By Lemma~\ref{LemmaHarmonicLimit}, if $U_n$ is uniform on $[n]$,
then $H_n(U_n)$ converges in distribution to
$Z\sim\operatorname{Exp}(1)$. Since
\[
    z\mapsto
    \Phi_{d+1}
    \left(
        \ell_m+\min\{z,\kappa_m\}
    \right)
\]
is bounded and continuous, its expectations converge as well. Since there are only finitely many pairs $(d,m)$, for every $\xi>0$ we may choose $K=K(M,\xi)$ such that, whenever $n_{md}\geq K$,
\[
    \frac{1}{n_{md}}
    \sum_{i\in I_{md}}
    \Phi_{d+1}
    \left(
        \ell_m+\widehat X_{id}
    \right)
    \leq
    \mathbb{E}\left[
        \Phi_{d+1}
        \left(
            \ell_m+\min\{Z,\kappa_m\}
        \right)
    \right]
    +\xi.
\]
By $\kappa_m=c(D-\ell_m)$ and
\eqref{eq:PhiRecursionAppendix}, the expectation on the right-hand side is
$\Phi_d(\ell_m)$. Combining this with
\eqref{eq:SlabLipBound} gives
\[
    \sum_{i\in I_{md}}
    \Phi_{d+1}(Y_{id})
    \leq
    n_{md}\Phi_d(\ell_m)
    +
    n_{md}
    \left(
        \eta+\xi+\frac{n_{md}}{B}
    \right).
\]
Since $\Phi_d$ is non-decreasing and
$Y_{i(d-1)}\geq\ell_m$,
\[
    n_{md}\Phi_d(\ell_m)
    \leq
    \sum_{i\in I_{md}}
    \Phi_d(Y_{i,d-1}).
\]
Therefore, every slab of size at least $K$ satisfies
\begin{equation}\label{eq:LargeSlabBound}
    \sum_{i\in I_{md}}
    \Phi_{d+1}(Y_{id})
    \leq
    \sum_{i\in I_{md}}
    \Phi_d(Y_{i(d-1)})
    +
    n_{md}
    \left(
        \eta+\xi+\frac{n_{md}}{B}
    \right).
\end{equation}

It remains to handle slabs with $n_{md}<K$. There are at most $M$ such
slabs, so they contain fewer than $MK$ advertisers in total. Since
$0\leq\Phi_d\leq D$, their total contribution to the left-hand side is at
most $DMK$.

Summing \eqref{eq:LargeSlabBound} over the large slabs, adding the small
slabs, and using
$\sum_m n_{md}=N$ and
$\sum_m n_{md}^2\leq N^2$, we obtain
\[
    \sum_{i \in [N]}
    \Phi_{d+1}(Y_{id})
    \leq
    \sum_{i \in [N]}
    \Phi_d(Y_{i,d-1})
    +
    N(\eta+\xi)
    +
    \frac{N^2}{B}
    +
    DMK.
\]
\end{proof}

\section{Additional Details for the Proof of Theorem~\ref{Thm3}}
\label{Appendix_Lemma41}

We complete the lower-bound calculation by showing that
\[
    \Gamma_{f^*}
    =
    \inf_{\substack{
        x\in[0,2],\;
        o_1,o_2\in[0,1]\\
        o_1+o_2>0
    }}
    \frac{
        F^*(x)
        +o_1\bigl(1-f^*(x\mid2,1)\bigr)
        +o_2\bigl(1-e^{x-2}\bigr)
    }{
        o_1+o_2
    }
    =
    1-e^{-h^*},
\]
where
\[
    f^*(x\mid2,1)
    :=
    \begin{cases}
        e^{x-h^*},
        & 0\leq x\leq2-h^*,\\[1mm]
        (3-x)e^{x-2},
        & 2-h^*<x\leq2,
    \end{cases}
\]
and
\[
    F^*(x):=\int_0^x f^*(u\mid2,1)\,du.
\]

We first show that
\[
    \Gamma_{f^*}
    =
    \inf_{\substack{x\in[0,2]\\ \lambda\in[0,1]}}
    \left\{
        F^*(x)\max\{\lambda,1-\lambda\}
        +\lambda\bigl(1-f^*(x\mid2,1)\bigr)
        +(1-\lambda)\bigl(1-e^{x-2}\bigr)
    \right\}.
\]
To derive an equivalent formulation of \eqref{NewFRP}, set
\[
    \lambda:=\frac{o_1}{o_1+o_2}\in[0,1].
\]
Since $o_1,o_2\leq1$,
\[
    \frac{1}{o_1+o_2}
    \geq
    \max\{\lambda,1-\lambda\}.
\]
Therefore, \eqref{NewFRP} implies
\[
    \Gamma_{f^*}
    \geq
    \inf_{\substack{x\in[0,2]\\ \lambda\in[0,1]}}
    \left\{
        F^*(x)\max\{\lambda,1-\lambda\}
        +\lambda\bigl(1-f^*(x\mid2,1)\bigr)
        +(1-\lambda)\bigl(1-e^{x-2}\bigr)
    \right\}.
\]

Now, for each fixed $x \in [0,2]$, the expression inside the braces is affine in $\lambda$ on $[0,1/2]$ and $[1/2,1]$. Choosing $(o_1,o_2)=(0,1)$, $(1,1)$, and $(1,0)$ gives in \eqref{NewFRP}, respectively,
\[
    F^*(x)+1-e^{x-2},
    \qquad
    \frac{
        F^*(x)+2-f^*(x\mid2,1)-e^{x-2}
    }{2},
    \qquad
    F^*(x)+1-f^*(x\mid2,1),
\]
which are exactly the values of the above expression at
$\lambda=0$, $\lambda=1/2$, and $\lambda=1$, getting
\[
    \Gamma_{f^*}
    \leq
    \inf_{\substack{x\in[0,2]\\ \lambda\in[0,1]}}
    \left\{
        F^*(x)\max\{\lambda,1-\lambda\}
        +\lambda\bigl(1-f^*(x\mid2,1)\bigr)
        +(1-\lambda)\bigl(1-e^{x-2}\bigr)
    \right\}.
\]

Also, note that, for fixed $x$, its minimum is therefore attained at
$\lambda\in\{0,1/2,1\}$.

Now, to prove $\Gamma_{f^*}\geq1-e^{-h^*}$, it suffices to verify, for every $x\in[0,2]$,
\[
\begin{aligned}
    F^*(x)+1-e^{x-2}
    &\geq 1-e^{-h^*},
    &&\text{for }(o_1,o_2)=(0,1),\\
    F^*(x)+2-f^*(x\mid2,1)-e^{x-2}
    &\geq 2(1-e^{-h^*}),
    &&\text{for }(o_1,o_2)=(1,1),\\
    F^*(x)+1-f^*(x\mid2,1)
    &\geq 1-e^{-h^*},
    &&\text{for }(o_1,o_2)=(1,0).
\end{aligned}
\]
Note that a direct integration yields
\[
    F^*(x)
    =
    \begin{cases}
        e^{-h^*}(e^x-1),
        & 0\leq x\leq2-h^*,\\[1mm]
        (4-x)e^{x-2}-2e^{-h^*},
        & 2-h^*<x\leq2.
    \end{cases}
\]

Suppose first that $0\leq x\leq2-h^*$. Direct substitution gives
\[
\left[
\begin{aligned}
    \qquad &F^*(x)+1-f^*(x\mid2,1)
    =
    1-e^{-h^*},\\
    &F^*(x)+1-e^{x-2}
    =
    1-e^{-h^*}
    +e^x\bigl(e^{-h^*}-e^{-2}\bigr)
    \geq
    1-e^{-h^*},\\
    &F^*(x)+2-f^*(x\mid2,1)-e^{x-2}
    =
    1-e^{-h^*}+1-e^{x-2}
    \geq
    2(1-e^{-h^*}).
\end{aligned}
\right.
\]
The inequalities follow from $h^*<2$ and $x\leq2-h^*$.

Next, suppose that $2-h^*<x\leq2$. Then
\[
\left[
\begin{aligned}
    \qquad & F^*(x)+1-f^*(x\mid2,1)
    =
    1-e^{-h^*}
    +\bigl(e^{x-2}-e^{-h^*}\bigr)
    \geq
    1-e^{-h^*},\\
    & F^*(x)+1-e^{x-2}
    =
    1-e^{-h^*}
    +\bigl(f^*(x\mid2,1)-e^{-h^*}\bigr)
    \geq
    1-e^{-h^*},\\
    &F^*(x)+2-f^*(x\mid2,1)-e^{x-2}
    =
    2(1-e^{-h^*}).
\end{aligned}
\right.
\]
Here, $e^{x-2}\geq e^{-h^*}$ and
$f^*(x\mid2,1)=(3-x)e^{x-2}\geq e^{-h^*}$ on this interval.

Conversely, taking $x=0$ and $(o_1,o_2)=(1,0)$ in \eqref{NewFRP} gives
\[
    \Gamma_{f^*}
    \leq
    F^*(0)+1-f^*(0\mid2,1)
    =
    1-e^{-h^*}.
\]
\end{APPENDICES}
\end{document}